\documentclass[a4paper]{article}
\usepackage[a4paper,margin=25mm]{geometry}

\usepackage{booktabs} 
\usepackage[backend=biber, style=apa, sortcites=true]{biblatex}
\DeclareDelimFormat{nameyeardelim}{\addcomma\space}
\usepackage{graphicx}
\usepackage{physics}
\usepackage{enumerate}
\usepackage{placeins}
\usepackage{amsmath}
\usepackage{amssymb}
\usepackage{amsthm}
\usepackage{comment}
\usepackage{arydshln}
\usepackage{bm} 
\usepackage{indentfirst}  
\theoremstyle{plain}
\newtheorem{thm}{Theorem}
\newtheorem{thmintro}{Theorem}
\newtheorem*{thm*}{Theorem}
\newtheorem{prop}{Proposition}
\newtheorem*{prop*}{Proposition}
\newtheorem{dfn}{Definition}
\newtheorem*{dfn*}{Definition}
\theoremstyle{definition}
\newtheorem{rmk}{Lemma}
\newtheorem*{rmk*}{Lemma}
\newtheorem{cor}{Corollary}
\newtheorem*{cor*}{Corollary}
\usepackage{tikz} 
\usetikzlibrary{3d}
\usetikzlibrary{positioning, intersections, calc, arrows.meta,math,patterns}

\begin{document}
\title{Selling Multiple Items to a Unit-Demand Buyer\\via Automated Mechanism Design\footnote{We would like to thank Michihiro Kandori for his guidance and comments. We are also grateful to Itai Ashlagi, Nima Haghpanah, Jason D. Hartline, Shengwu Li, Elliot Lipnowski, Paul Milgrom, Ryuji Sano, Ilya Segal, Alex Teytelboym, Frank Yang, Donghao Zhua as well as participants at the 7th World Congress of the Game Theory Society, Beijing, China and 2025 Asian School in Economic Theory (formerly the Summer School of Econometric Society), Seoul, South Korea.}}
\author{
Kento Hashimoto\thanks{Faculty of Economics, The University of Tokyo. Email: academic.research.purpose@gmail.com},
Keita Kuwahara\thanks{Graduate School of Economics, The University of Tokyo. Email: kuke0303@g.ecc.u-tokyo.ac.jp},
and Reo Nonaka\thanks{Department of Economics, Yale University. Email: lemonnade@g.ecc.u-tokyo.ac.jp}
}

\date{\today}

\maketitle

\begin{abstract}   
    Finding the optimal (revenue-maximizing) mechanism to sell multiple items has been a prominent and notoriously difficult open problem. Existing work has mainly focused on deriving analytical results tailored to a particular class of problems (for example, \cite{giannakopoulos2015bounding,yang2025nested}). The present paper explores the possibility of a generally applicable methodology of the Automated Mechanism Design (AMD). We first employ the deep learning algorithm developed by \textcite{dutting2023optimal} to numerically solve small-sized problems, and the results are then generalized by educated guesswork and finally rigorously verified through duality. By focusing on a single buyer who can consume one item, our approach leads to two key contributions: establishing a much simpler way to verify the optimality of a wide range of problems and discovering a completely new result about the optimality of uniform pricing. First, we show that selling each item at an identical price (or equivalently, selling the grand bundle of all items) is optimal for any number of items when the value distributions belong to a class that includes the uniform distribution as a special case. Different items are allowed to have different distributions. Second, for each number of items, we established necessary and sufficient conditions that $c$ must satisfy for uniform pricing to be optimal when the value distribution is uniform over an interval $[c, c+1]$. This latter model does not satisfy the previously known sufficient conditions for the optimality of grand bundling \parencite{haghpanah2021pure}. Our results are in contrast to the only known results for $n$ items (for any $n$), \textcite{giannakopoulos2015bounding} and \textcite{daskalakis2017strong}, which consider a single buyer with additive preferences, where the values of items are narrowly restricted to i.i.d. according to a uniform or exponential distribution.
\end{abstract}

\section{Introduction}
The optimal auction design problem is an important yet notoriously difficult open question in economic theory. It considers the revenue-maximizing selling procedure when buyers’ willingness to pay is private information. The seminal work \textcite{myerson} solved the problem of selling a single item, but for more than 40 years, the optimal mechanism for multiple items has not been known. Even when there is one buyer, we are yet to obtain a general result. The only known results for the one-buyer and arbitrary number of items (for any $n$-item) case are \textcite{giannakopoulos2015bounding}, \textcite{daskalakis2017strong}, \textcite{haghpanah2021pure}, and \textcite{yang2025nested}. \textcite{giannakopoulos2015bounding} and \textcite{daskalakis2017strong} considered a buyer with additive preferences, where the value of a set of items is the sum of the values of each item and derived the optimal mechanism when the values of items are i.i.d. according to a particular distribution.   \textcite{haghpanah2021pure} covers broader settings but relies on conditions that are not easy to check in practice. Similarly, \textcite{yang2025nested} provides a general result, but their approach is limited to a one-dimensional type space, reducing its applicability to more complex multi-dimensional scenarios. In summary, results that apply to cases with an arbitrary number of goods are quite rare, and even when they exist, they hold only under limited settings. Furthermore, although they show sufficiency conditions for a specific mechanism to be optimal, it remains unsolved whether those conditions are also necessary. The necessary condition \textcite{haghpanah2021pure} showed is partial inverse. Although \textcite{daskalakis2017strong} showed a necessary and sufficient condition for grand bundling to be optimal, that condition is defined by the notion of stochastic dominance and is not easy to verify.\footnote{ Indeed, even when the value of each item uniformly distributed over $[c, c+1]$ independently, \textcite{daskalakis2017strong} only prove the sufficient condition for the grand bundling to be optimal and does not show necessary and sufficient characterization for what value of $c\ge0$, grand bundling is optimal. In contrast, focusing on the case the buyer has unit-demand preference, we characterized that condition and that is quite easy to check. }  A necessary and sufficient characterization in multiple goods auction theory is significantly rare and even if they exist, the condition is not easy to verify. 
\par
In contrast, we derive thus far the most general results in terms of the number of items and the range of possible value distributions and established easy-to-verify necessary and sufficient condition for the uniform pricing to be optimal  by focusing on a \textit{unit-demand} buyer who can consume at most one item. We have two contributions. The first theorem characterizes the conditions under which bundling (selling all items together) is optimal, regardless of the item types or the values of $\alpha$. The required condition, quantile-scaled monotonicity, is satisfied by a broad class of distributions, including the uniform and beta distributions. we assume that the cumulative distribution functions (CDFs) for the values of items $j=1,\ldots, n$ are given by $F_j(x_j)=F(x_j)^{\alpha_j}$ for $\alpha_j > 0$ and some base CDF $F(x)$. Here, the parameters $\alpha_1,\ldots, \alpha_n$ may be distinct. This setup represents a scenario where all items share a common underlying distribution but differ specifically by their $\alpha$ exponents. Second, we find easy-to-verify necessary and sufficient conditions for uniform pricing to be optimal when the value distribution is uniform. Suppose each item's value distribution is uniform over $[c,c+1]$. For any number of items, the uniform pricing is optimal if and only if $c$ is no greater than a certain threshold, which is uniquely determined according to the number of items. We succeeded in showing a necessary and sufficient characterization of the optimality. \par
 Our results showcase the Automated Mechanism Design (AMD) approach developed in computer science \parencite{feng2018deep,conitzer2002complexity}
, which is not yet well-known and appreciated by economic researchers. Mechanism design problems, in general, are constrained optimization problems, but analytical solutions are often hard to find.  This comes from the complexity of incentive constraints, especially for multi-dimensional private information, and this is also true for the optimal auction problem with multiple items. The Automated Mechanism Design approach numerically solves the optimal mechanism using a computer algorithm for small-sized problems for specific parameter values, and the results obtained are generalized. We employ the deep-learning-based algorithm developed by \textcite{dutting2023optimal}, which is known to reproduce the optimal auction mechanisms in the literature. We apply it to a two-item or three-item case with uniform value distribution and obtain a candidate optimal mechanism. Using insights from that particular instance of the problem, we generalize the result and rigorously prove optimality using our duality approach. 
Our basic approach is described as follows. \textcite{myerson} provided an elegant procedure for deriving a simple closed form of the expected revenue from the incentive constraints. Still, this approach is challenging to extend to the multiple-item case. The difficulty comes from the fact that the private information is multi-dimensional. In ex-ante,  we do not know in which direction the incentive constraints bind. In contrast, we derive the optimal mechanism by duality. We adopt the approach proposed by \textcite{haghpanah2021pure}, which demonstrates that uniform pricing is optimal if the profile of relative values is stochastically non-decreasing in the positive values of the grand bundle. Our findings reveal that if the cumulative distribution function (CDF) F satisfies a specific property, the condition required by \textcite{haghpanah2021pure} is consistently satisfied for any arbitrary number of goods. For the case where the value distribution is uniformly distributed over $[c, c+1]$, we utilized the strong duality of the dual problem, first established by \textcite{daskalakis2017strong}. \textcite{daskalakis2017strong} demonstrated the strong duality between the optimal mechanism design problem and its dual problem in the setting where the buyer has additive preferences. \textcite{kash2016optimal} extended this strong duality framework to the case where the buyer has unit-demand preferences. We derived the solution to this dual problem and verified its feasibility. Furthermore, by fully leveraging this strong duality framework, we established the necessary conditions for uniform pricing to be optimal. Combined with the sufficient conditions, these conditions provide the necessary and sufficient conditions for pure-bundling to be optimal.  
 
\subsection{Our results}
  In this paper, we derive the most general results regarding the number of items and the range of possible value distributions by focusing on a \textit{unit-demand} buyer who can consume no more than one item. 
In Section \ref{sec:when's alpha}, we assume that the cumulative distribution functions (CDFs) for the values of items $j=1,\ldots, n$ are given by $F_j(x_j)=F(x_j)^{\alpha_j}$ for $\alpha_j > 0$ and some base CDF $F(x)$. Here, the parameters $\alpha_1,\ldots, \alpha_n$ may be distinct. This setup represents a scenario where all items share a common underlying distribution but differ specifically by their $\alpha$ exponents.

We demonstrate that setting an identical price for each item—or equivalently, selling the grand bundle of all items—is optimal for any set of parameters $\alpha_1, \ldots, \alpha_n$:

\begin{thmintro}
    Assume $F$ is scale monotone. Then for any number of items and for any $\alpha_1, \ldots, \alpha_n>0$, uniform pricing is optimal.
\end{thmintro}

This theorem characterizes the conditions under which bundling (selling all items together) is optimal, regardless of the item types or the values of $\alpha$. Because this holds for an arbitrary number of items for any fixed $\alpha$, the result is applicable even in small-scale cases, such as a two-item setting. The required condition, scale monotone, is satisfied by a broad class of distributions, including the uniform and beta distributions.

In our second result, we provide necessary and sufficient conditions for the optimality of uniform pricing under the assumption that all bidders’ valuations are i.i.d. according to a uniform distribution on 
$[c,c+1]$. We identify a unique threshold $c^*$, which is determined by solving an $n$-dimensional equation that precisely indicates whether uniform pricing is optimal.
\begin{thmintro} \label{thmintro:U[c, c+1]}
Suppose the value of each item is uniformly distributed over $[c, c+1]$ identically and independently. Then
$$\text{The optimal mechanism is uniform pricing if and only if }\;  c \in[0, c^*]$$
where $c^*\in(0, n)$ is determined as follows: given an optimal price of grand bundle $p$ (denote it as $p(c)$ since it depends on $c$), the solution of the equation $(n+1)(c+1-p(c)) = c+1$ on $c\in(0, n)$ is uniquely determined. That is $c^*$.
\end{thmintro}
This result was also conjectured by AMD, but our proof strategy leverages strong duality; rather than relying on weak duality, we use an approach proposed by \textcite{daskalakis2017strong} and apply it to the unit-demand case following \textcite{kash2016optimal}. The key to the proof is the ``pushed measure" we define in Section \ref{sec:strong dual}. 
For sufficiency, we show that uniform pricing is optimal if the pushed measure convexly dominates the original measure (in a specific sense). This measure-theoretic relation holds when the value of \( c \) is smaller than the threshold \( c^* \). 
For necessity, rather than constructing an alternative mechanism, we demonstrate that assuming \( c > c^* \) leads to a contradiction in the optimality conditions. There are two key conditions for optimality, but these cannot simultaneously hold when \( c > c^* \).
\subsection{Related work}
Our work belongs to a large body of literature on multi-item optimal auction design.  There are very few results for the multiple-buyers case \parencite{palfrey1983bundling,yao2017,kolesnikov2022beckmann}, and the existing literature primarily focuses on the case of a single-buyer \parencite{pavlov2011optimal,daskalakis2017strong,haghpanah2021pure}. The general optimal mechanism for the single buyer case is unknown, but some specific cases are solved. \textcite{manelli} established the condition under which the price schedule mechanism is optimal when the buyer has an additive preference and solved for two items case with some assumption on value distributions. \textcite{pavlov2011optimal} derived the optimal mechanism for two items when the buyer's preference is either additive or unit-demand and the value of those uniformly distributed on $[c, c+1]$ for any $c\ge0$. 
\textcite{thirumulanathan2019optimal} extended this result for two items, unit-demand buyer case, and established the optimal mechanism when the value of each item uniformly and independently distributed over $[c, c+d_1]$ and $[c, c+d_2]$ respectively. There are only a few results that proved to be optimal for $n$ items (for any $n$); they are \textcite{giannakopoulos2015bounding}, \textcite{daskalakis2017strong}, \textcite{haghpanah2021pure} and \textcite{yang2025nested} to best of our knowledge. \textcite{giannakopoulos2015bounding} and \textcite{daskalakis2017strong} consider the case in which the buyer’s valuation is additive, and they assume that the values of items are i.i.d. according to a specific value distribution.  In particular, the former assumes the exponential distribution, where the density of value $x$ is $f(x)=\lambda e^{-\lambda x}$, while the latter assumes the uniform distribution over $[c, c+1]$, where $c$ depends on the number of items. \footnote{Although \textcite{giannakopoulos2015bounding} considers the exponential distribution for any $\lambda>0$, it essentially covers a single distribution in the following sense. Consider the problem with any $\lambda$. If we change the unit of account so that the original $\lambda$ is now called $1$ in the new representation, the problem boils down to the one with $\lambda=1$. Hence, all models can be expressed as the model under a single distribution with $\lambda=1$. In contrast, the distributions we allow remain distinct even after the changes in the unit of account.}  The optimal mechanism in both cases is to sell a bundle of all items. \textcite{haghpanah2021pure} provides a sufficient condition for selling all items at once to be optimal, but this condition is not easy to verify. \textcite{yang2025nested} provided the sufficient condition for nested bundling to be optimal, but they focused on the case when the type space is one-dimensional.

\section{Model}

We consider a mechanism for one buyer and $n\ge2$ items $N=\qty{1,\ldots,n}$. The buyer has a valuation function $t:2^N\rightarrow \mathbb{R}_{\geq 0}$, where $t(b)$ denotes how much the buyer values the subset of items $b\in B\equiv 2^N$. In particular, we consider the buyer to have a \textit{unit-demand} valuation, where $t$ satisfies $t(b)=\max_{j\in b} t(\qty{j})$ (and $t(\emptyset)=0$). Intuitively, the buyer cannot consume more than one item and chooses one that provides the largest utility from $b$. In this case, $n$ numbers $t(\qty{1}),\ldots, t(\qty{n})$ fully characterizes the valuation function $t$. Therefore, it suffices to consider the buyer to have $n$ values $(t_1,\ldots, t_n)\in \mathbb{R}_{\geq 0}^n$. 
For each $t$, the valuation function is given by
\[
v(b,t) = \max_{j \in b} t_j ,
\]
with $v(\emptyset,t) = 0$.
Let $b^* = N$ be the grand bundle and
\(
v(b^*,t) \in \mathbb{R}
\)
denote the value of the grand bundle for $t$.
\(
r(\cdot,t) \in \mathbb{R}^B
\)
denotes the profile of relative values of $t$, defined by
\[
r(b,t) = \frac{v(b,t)}{v(b^*,t)} \quad \text{for all } b \in B,
\]
whenever $v(b^*,t) > 0$.
\par

We consider $t_j$ to be in interval $T_j=[\underline{v},\overline{v}]\subset  \mathbb{R}_{\geq 0}$. We denote the hyperrectangle of all possible buyer values by $T=\prod_{j=1}^n T_j$. The seller knows some probability distribution over $T$.
Assume that valuations are independent across items $j$. For each item $j$, the valuation is distributed according to a cumulative distribution function (CDF)
$F_j : [\underline{v}, \overline{v}] \to [0,1]$, with a strictly positive density $f_j$ on $(\underline{v}, \overline{v})$. 
Consequently, $F_j$ is continuous and strictly increasing on $[\underline{v}, \overline{v}]$.
\par

The revelation principle \parencite{mas1995microeconomic} shows that the optimal mechanism can be found in a class of mechanisms where the buyer reports her valuations and a truthful report is optimal. Hence, we denote our mechanism $(a,p)$ as a pair of allocation rules $a:T\rightarrow \Delta(N\cup\qty{0})$ and payment rules $p:T\rightarrow \mathbb{R}_{\geq 0}$. we consider the probabilistic allocation over $n+1$ candidates: allocating $\qty{1},\qty{2},\ldots, \qty{n}$, or allocating nothing. We do not consider allocating more than two items because the buyer consumes only one item. The buyer's utility when reporting $t^{\prime}\equiv(t_1^{\prime},\ldots, t_n^{\prime})$ under true valuation $t\equiv(t_1,\ldots, t_n)$ is \begin{gather*}
u(t^{\prime};t)=\sum^n_{j=1}a_j(t^{\prime})t_j-p(t^{\prime})
\end{gather*}
The seller's utility is captured by the revenue given by the mechanism $p(t^{\prime})$.\par
It is standard in mechanism design to require mechanisms to satisfy the following properties:\\
{\textit{Individual rationality (IR)}:} $u(t;t)\geq 0$ for all $t\in T$\\
{\textit{Incentive compatibility (IC)}:} $u(t;t)\geq u(t^{\prime};t)$ for all $t,t^{\prime}\in T$
We use the following characterization of incentive-compatible mechanisms by \textcite{rochet}. 
\begin{rmk}[\textcite{rochet}]\label{lem:truthful auction}
A mechanism $M = (a, p)$ is truthful (IC) if and only if the utility functions $u(t)\equiv u(t|t)$ which the mechanism $M$ induces have the following properties:
    (i) $u$ is a convex function.
   (ii) $u$ is almost everywhere (a.e.) differentiable with
    $\pdv{u(t)}{t_j} = a_j(t) $
for all items $j=1,\ldots,n$ and a.e. $t \in T.$
\end{rmk}
We study the problem of maximizing the seller's expected revenue based on their prior knowledge of the joint distribution $F$ under the IR and IC constraints. By lemma \ref{lem:truthful auction},
it suffices to solve
\begin{gather}\label{problem:primal}
\sup_{u} \int_{T} (\nabla u(\boldsymbol{t})\cdot \boldsymbol{t} - u(\boldsymbol{t}))f(\boldsymbol{t}) d\boldsymbol{t}
\end{gather}
over the space of nonnegative convex functions $u$ on $T$
 having the properties
$$\sum^n_{j=1}\pdv{u(\boldsymbol{t})}{t_j}\leq 1\text{ and }
\pdv{u(\boldsymbol{t})}{t_j}\geq 0\text{ for a.e. $\boldsymbol{t}\in T$, all $j \in \qty{1,...,n} $.}$$

In particular, we demonstrate the optimality of \emph{uniform pricing}. Uniform pricing is a mechanism in which the seller posts an identical price for all goods. Let $p$ denote this price. The utility function of a unit-demand buyer is given by
\[
    u(t) = u^*(t) \equiv \max\left\{ \max(t) - p, 0 \right\},
\]
where $\max(t)$ denotes the maximal component of the vector $t$, i.e.,
\[
    \max(t) = \max\left\{ t_1, \ldots, t_n \right\}.
\]
Hereafter, we denote this specific utility function by $u^*(t)$. Under the unit-demand assumption, uniform pricing is equivalent to a mechanism that offers only the grand bundle at price $p$.

\section{The results of general distribution}\label{sec:when's alpha}

In this section, we generalize the optimality of uniform pricing to a broader class of distributions. We adopt the framework proposed by \textcite{haghpanah2021pure}, which establishes that grand bundling is optimal if the profile of relative values is stochastically non-decreasing in the positive values of the grand bundle. Our main contribution is identifying a specific property of the distribution $F$ that guarantees this condition is satisfied for any arbitrary number of items and various power-scale parameters.

First, we introduce the concept of scale-monotonicity for a cumulative distribution function (CDF), which plays a pivotal role in our analysis.

\begin{dfn}
A CDF $F:[\underline{v}, \overline{v}]\to [0,1]$ is scale-monotone if,
for all $\omega \in (\underline{v}/\overline{v},1)$, the function
\[
x \mapsto \frac{F(\omega x)}{F(x)}
\]
is non-increasing in $x \in (\underline{v}/w, \overline{v}]$.
\end{dfn}

Based on this definition, we show that scale-monotonicity is a sufficient condition for the optimality of uniform pricing, even when items are not identically distributed but follow power-scale transformations of the same base distribution.

\begin{thm}\label{thm:when's alpha}
    Assume $F$ is scale-monotone. For each $i=1, \ldots, n$, take any $\alpha_i>0$ and consider the CDF $F^{\alpha_i}$ over $[\underline{v}, \overline{v}]$. Suppose $t_i$ follows CDF $ F^{\alpha_i}$ independently. Then for any number of items and for any $\alpha_1, \ldots, \alpha_n>0$, uniform pricing is optimal.
\end{thm}

As a direct consequence of this theorem, we obtain the following result for the standard i.i.d. case.

\begin{cor}\label{cor: i.i.d.}
    Assume $F$ is scale-monotone. Consider the setting where $t_i\overset{i.i.d.}{\sim} F$. Then for any number of items, uniform pricing is optimal.
\end{cor}

Furthermore, we can relate this to the condition in \textcite{haghpanah2021pure}. Specifically, if their condition for uniform pricing optimality holds for a base case of two i.i.d. items, the scale-monotonicity ensures the robustness of this optimality as we scale the number of items or shift the distributions.

\begin{cor}\label{cor:two item}
    Suppose that the profile of relative values is stochastically non-decreasing in positive values of the grand bundle under the setting with $n=2, t_1, t_2\overset{i.i.d.}{\sim}F$. Then for any $n\ge2$ and $\alpha_1, \ldots, \alpha_n>0$, uniform pricing remains to be optimal when the seller offers \( n \) items and $t_i$ follows CDF \( F^{\alpha_i} \) for each $i=1, \ldots, n$.
\end{cor}

To provide a more tractable criterion for scale-monotonicity, we examine the elasticity of the density function. The following condition is often easier to verify for standard parametric distributions.

\begin{dfn}
    Assume density function $f$ is differentiable. We say $f$ satisfies the monotone elasticity condition if the elasticity of $f$, 
    $$\frac{tf'(t)}{f(t)}$$
    is non-decreasing in $t$.
\end{dfn}

Finally, we establish the connection between the monotone elasticity of the density and the scale-monotonicity of the CDF, providing a practical tool for applying our main results.

\begin{prop}\label{prop:hashimoto no zyuubunn zyoukenn}
    Assume density function $f$ is differentiable and satisfies the monotone elasticity condition. Then its CDF $F$ is scale-monotone.
\end{prop}

\section{The results of $U[c, c+1]$ cases}\label{sec:strong dual}
\subsection{Results by Strong Duality}

In this section, we find the optimal mechanism for selling $n$ items to a unit-demand buyer utility function for broader cases. We assume each valuation identically and independently follows the distribution $U[c,c+1](c\geq 0)$. Let $G$ be a cumulative distribution function of $U[c,c+1]$ and $p$ be a unique maximizer of expected revenue when we sell a grand bundle with price $q\in[c, c+1]$: $q(1-G(q)^n)$. \footnote{$q(1-G(q)^n)$ has a unique maximizer on $(c, c+1)$. To see this, let $h(q)=(q(1-G(q)^n))^{\prime}$. We can prove the existence of the unique maximizer on $(c, c+1)$ by getting $h(c)=1>0$, $h(c+1)=-n(c+1)<0$ and $h'(q)<0$ for all $q\in (c,c+1)$. }

\begin{thm} \label{thm:U[c, c+1]}
Suppose the value of each item is uniformly distributed over $[c, c+1]$ identically and independently. Then
$$\text{The optimal mechanism is uniform pricing if and only if }  c \in[0, c^*]$$
where $c^*\in(0, n)$ is determined as follows: given an optimal price of grand bundle $p$ (denote it as $p(c)$ since it depends on $c$), the solution of the equation $(n+1)(c+1-p(c)) = c+1$ on $c\in(0, n)$ is uniquely determined. That is $c^*$.

\end{thm}
We show the above theorem in Section \ref{subsec:Dual Problem for Strong Duality}, \ref{subsec:Pushed measure}, \ref{subsec:Proof Sketch}, and Appendix.

\subsection{Analysis of $c^*$}\label{subsec:analysis of c^star}
In this subsection, we will see some qualitative properties of $c^*$.
Since the value of $c^*$ depends on $n$, we will denote it as $c_n^*$ in this section. First of all, $c_n^*$ is well-defined and uniquely determined. For example, $c_2^*=1, c_3^*=\frac{11-\sqrt{33}}{4}\simeq 1.31, \text{ and } c_4^*\simeq1.57$.
\begin{prop}\label{prop:uniqueness of c_n^*}
$n$-dimensional equation $(n+1)^n-(n-c)^n-n^2(c+1)(n-c)^{n-1}=0$ has a unique solution on $c\in(0, n)$. That is  $c_n^*$.
\end{prop}
\begin{proof}
    The proof is in Appendix \ref{proof:analysis of c^star}.
\end{proof}
By Proposition \ref{prop:uniqueness of c_n^*} and its proof, the following corollary can be derived.
\cor{
Suppose the value of each item is uniformly distributed over $[c, c+1]$ identically and independently.
The optimal mechanism is uniform pricing if and only if
\[
c\in [0,n] \text{ and } (n+1)^n-(n-c)^n-n^2(c+1)(n-c)^{n-1}\leq 0
\]
}
Thus, given number of items $n\ge2$ and the value $c\ge0$, we can easily verify whether the uniform pricing is optimal by above condition.

In particular, this threshold is strictly increasing with respect to $n$, meaning that uniform pricing becomes more likely to be optimal as the number of items increases. Moreover, as the number of items approaches infinity, uniform pricing becomes optimal for any $c\ge0$, since the threshold diverges to infinity as $n$ goes to infinity. We will formally state these results as follows:

\begin{prop}\label{prop:c^star strictly increasing}
    $c_n^*$ is strictly increasing in $n\ge2$.
\end{prop}
\begin{proof}
    The proof is in Appendix \ref{proof:analysis of c^star}.
\end{proof}

\begin{prop}\label{prop:c_n infty}
$c_n^*\to\infty$  as $n\to\infty $. In particular, we have $c_n^*>\frac{\ln n}{3}$ for all $n\ge2$.
\end{prop}
\begin{proof}
    The proof is in Appendix \ref{proof:analysis of c^star}.
\end{proof}

The interpretation of the above two results is as follows: Fixing any $c\ge0$, consider the case when the monopolist guesses that the buyer's value is distributed according to $U[c, c+1]$ independently. When there are not enough goods, the optimal mechanism may not be uniform pricing. But if we increase the number of items, the probability that the highest value among the goods is close to $c+1$ increases; more formally, the expected value of the highest order statistic gets close to $c+1$. Therefore, rather than presenting a mechanism with unnecessarily additional menus, offering a mechanism with only the grand bundle allows for greater surplus extraction (recall that we only need to consider mechanisms that allocate at most one good to the buyer since the buyer has a unit-demand preference.)

\subsection{Dual Problem for Strong Duality}\label{subsec:Dual Problem for Strong Duality}
Below, we sketch the main ideas of the proof of Theorem \ref{thm:U[c, c+1]}.
Instead of using the weak duality discussed throughout this paper, we switch to strong duality in this section.
\textcite{daskalakis2017strong} introduced the strong duality in the additive case, and \textcite{kash2016optimal} extended it to a unit-demand case. 
In this section, following \textcite{daskalakis2017strong}, we denote the buyer's type by $x \in T$, and let $X(=T)$ represent the set of all such types.

After defining some components, we state the strong duality. We introduce some useful  notations in the same way as in \textcite{daskalakis2017strong}:
\begin{itemize}
  \item Denote $\mathcal{U}(X)$ as the set of all continuous, non-decreasing, and convex functions $u:X\to\mathbb{R}$
  \item $\Gamma(X)$ and $\Gamma_{+}(X)$ denote the sets of signed and unsigned measures on $X$ respectively.
  \item Given an unsigned measure $\gamma \in \Gamma_{+}(X \times X)$, $\gamma_{1}$ and $\gamma_{2}$ denote the two marginals of $\gamma$, i.e. $\gamma_{1}(E)=\gamma(E \times X)$ and $\gamma_{2}(E)=\gamma(X \times E)$ for all measurable sets $E \subset X$.
  \item For a (signed) measure $\mu$ and a measurable $E \subset X$, we define the restriction of $\mu$ to $E$, denoted $\left.\mu\right|_{E}$, by the property $\left.\mu\right|_{E}(S)=\mu(E \cap S)$ for all measurable $S$.
  \item For a signed measure $\mu$, we will denote by $\mu_{+}, \mu_{-}$the positive and negative parts of $\mu$, respectively. That is, $\mu=\mu_{+}-\mu_{-}$, where $\mu_{+}$and $\mu_{-}$provide mass to disjoint subsets of $X$.
\end{itemize}

\begin{dfn}[\textcite{daskalakis2017strong}, Definition 3]
    The transformed measure $\mu$ of $f$ is the signed measure supported within $X$ defined as
    $$\mu(E)\equiv \int_{\partial X}\mathbb{I}_E(x)f(x)(x\cdot\hat{n})dx-\int_X\mathbb{I}_E(x)(\nabla f(x)\cdot x+(n+1)f(x))dx+\mathbb{I}_E((L_1, \dots, L_n))$$
    for all measurable sets $E$, where $\mathbb{I}_E(\cdot)$ denotes the indicator function on $E$ and $\hat{n}$ denotes the outer unit normal field to the boundary $\partial X$.
\end{dfn}

We will calculate transformed measure of our setting: $n$ items uniformly and independently distributed over $[c, c+1]$ as in \textcite{daskalakis2017strong}, Subsection 2.2:\par

Suppose each value of $n$ items uniformly and independently distributed over $[c, c+1]]$ with $c\ge0$. So that $X=[c, c+1]^n$.
    
    The joint distribution of the items is given by the constant density function $f=1$ throughout $X$. The transformed measure $\mu$ of $f$ is given by the relation
    $$
    \mu(E)\equiv\mathbb{I}_{E}\left(c, \ldots, c\right)+ \int_{\partial X} \mathbb{I}_{E}(x)(x \cdot \hat{n}) d x-(n+1) \int_{X} \mathbb{I}_{E}(x) d x
    $$
    for all measurable set $E\subset X$. Therefore, $\mu$ is the sum of:
    \begin{itemize}
      \item A point mass of +1 at the point $\left(c, \ldots, c\right)$.
      \item A mass of $-(n+1)$ distributed uniformly throughout the region $X$.
      \item For each $i=1, \ldots, n$,  a mass of $+(c+1)$ distributed uniformly on each surface $\left\{x \in \partial X: x_{i}=c+1\right\}$.
      \item For each $i=1, \ldots, n$,  a mass of $-c$ distributed uniformly on each surface $\left\{x \in \partial X: x_{i}=c\right\}$.
    \end{itemize}
Then, as in \textcite{kash2016optimal}, we can rewrite the above optimization problem as
\begin{align}
    &\sup_{u\in\mathcal{U}(X)}\int_{X} u d\mu\label{eq:daskalakis}\\
    &\quad s.t.\quad u(x)-u(x')\le\|x-x'\|_\infty, \; \text{ for all } x, x'\in X.
\end{align}

\begin{dfn}[\textcite{daskalakis2017strong}, Definition 5]
    Let $\alpha$ and  $\beta$ be a signed Radon measure on $X$. $\alpha$ convexly dominates  $\beta$, denoted  $\alpha\succeq_{\text{cvx}}\beta$, if for all $u\in\mathcal{U}(X)$, $$\int_X ud\alpha \geq \int_X u d\beta$$
    Moreover, given a measurable set $E\subset X$, we denote $\alpha|_E\succeq_{\text{cvx}}\beta|_E$, if for all $u\in\mathcal{U}(X)$, 
    $$\int_E ud\alpha \geq \int_E u d\beta$$
\end{dfn}
The following lemma corresponds to the strong duality framework first proposed by \textcite{daskalakis2017strong}. The following is its unit-demand version, as introduced in \textcite{kash2016optimal}. While this is not technically strong duality in its original form, it serves as a helpful variant by focusing on complete slackness and directly formulating the optimality conditions.

\begin{rmk}(\textcite{kash2016optimal}, Corollary 4.1.) \label{lemma:Kash and Frongillo}
$u$ is optimal in \eqref{eq:daskalakis} if and only if there exists  a positive Radon measures supported within $X\times X$, $\gamma$, such that
\begin{align}
  \gamma_1-\gamma_2  &\succeq_\text{cvx} \mu  \label{eq:cvx-cdn}\\
  \int_{[c,c+1]^n}ud(\gamma_1-\gamma_2)&=\int_{[c,c+1]^n}ud\mu \label{eq:meas-cdn}\\
  u(x)-u(y)&=\|x-y\|_\infty~ \gamma(x,y)\text{-almost surely}\label{eq:norm-cdn}
\end{align}
where $\gamma_1$ and $\gamma_2$ are Radon measure on $X$ defined as $\gamma_1(E)=\gamma(E\times X)$ and $\gamma_2(E)=\gamma(X\times E)$.
\end{rmk}

Since we conjecture that uniform pricing is optimal, we set the function \( u: X \to \mathbb{R} \), introduced above, to be:  
\[
u(x) =u^*(x)= \max\qty{\max(x) - p, 0}.
\]
In the following subsection, we will identify a positive Radon measure \( \gamma \) on \( X \times X \) that satisfies the optimality condition in the sufficiency proof. For the necessity proof, we will show that no such measure \( \gamma \) can exist if \( u \) is defined as above. This establishes a clear distinction between the cases where uniform pricing is optimal and where it is not.

\subsection{Pushed measure}\label{subsec:Pushed measure}

The following subsection identifies a measure \( \gamma \) that ensures \( u \), associated with uniform pricing, remains optimal. Here, the concept of the ``pushed measure" plays a crucial role in the proof. As a preparatory step, we first define the upper set and introduce the terminology for each region in \( X \) that we divide.

\begin{dfn}
     Given $k\in\mathbb{N}$ and $E\subset \mathbb{R}^k$, we call $U\subset E$ is a upper set with respect to $E$ if
     $$\text{ For all }x\in U \text{ and }x^{\prime}\in E, \; x\leq x^{\prime} \text{ implies that } x^{\prime}\in U$$
     When $U\subset E$ is an upper set with respect to $E$, we would just say $U\subset E$ is an upper set with abuse of terminology.
\end{dfn}
Denote each region of $X$ as follows. 
\begin{gather*}
    W \equiv \qty{x\in X: \max(x)\in (p,c+1]}\\
    Z \equiv \qty{x\in X: \max(x)\in [c,p]}\\
    B\equiv \qty{x\in X: \max(x)=c+1}\\
    B_Z \equiv \qty{x\in X: \max(x)=p}
\end{gather*}
    
Now, we are ready to introduce the pushed set and the pushed measure.

\begin{dfn} \label{dfn:push U}
For upper set $U\subset B$, the pushed set $push(U)$ is defined on $W\cup B_Z$ as follows:
    \begin{gather*}
        push(U)\equiv \qty{y\in W\cup B_Z|  \text{ there exists } x\in U\cap B \text{ such that }  u^*(x)-u^*(y)=\|x-y\|_{\infty}}
    \end{gather*}
\end{dfn}

\begin{dfn}\label{dfn:gamma_1}
    $\gamma_1$ is a measure on $B$ defined as
\begin{gather*}
    \gamma_1(U)\equiv\mu_-(push(U) )
\end{gather*}
for all measurable upper set $U\subset B$.\footnote{Notice that $\gamma_1$ is well-defined as a measure on $B$ by above since Borel $\sigma$-algebra on $B$ is generated by the set of all measurable upper sets on $B$.}
\end{dfn}


In fact, the pushed measure plays a crucial role in establishing uniform pricing's optimality. This measure is particularly significant in understanding why uniform pricing remains the best mechanism under certain conditions. However, since its relevance is more directly tied to the necessity argument, we will postpone a detailed explanation until the necessity section, where we will provide a thorough discussion on why this pushed measure is essential for proving optimality. 

Here, instead of explaining more about the pushed measure, we will introduce an equivalent characterization of the pushed set defined in Definition \ref{dfn:push U} as follows:
\begin{rmk}\label{lemma:y+t1}
For all upper set $U\subset B$ with respect to $B$ and all $y\in W\cup B_Z$, the following holds\footnote{we will denote $\bm{1}$ as the $n$-dimensional vector in $\mathbb{R}^n$ whose components are all $1$, i.e. $\bm{1}\equiv(1, \dots, 1)$.}
    \begin{gather*}
        ~y+t\bm{1} \in U\cap B \text{ for some } t \in \mathbb{R}\iff y\in push(U)
    \end{gather*}
\end{rmk}
\begin{proof}
    The proof is in Appendix \ref{proof:uniform pricing is optimal if gamma_1>mu_+}.
\end{proof}

\FloatBarrier
\begin{figure}[htp]
\centering
\begin{tikzpicture}[scale=1.2]

  \draw[->, thick] (0,0) -- (4.5,0) node[right] {\small $x_1$};
  \draw[->, thick] (0,0) -- (0,4.5) node[above] {\small $x_2$};

  \fill[green!50!white, opacity=0.5] (0,2) -- (2,2) -- (2,0) -- (4,0) -- (4,4) -- (0,4) -- cycle;
  
  \fill[gray!50!white, opacity=0.7] (4,4) -- (3,4) -- (1,2) -- (2,2) -- (2,0.2) -- (4,2.2) -- cycle;
  
  \draw[line width=2.5pt, black] (3,4) -- (4,4) -- (4,2.2);

  \draw[line width=2pt, orange!80!black] (0,2) -- (2,2) -- (2,0);

  \draw[line width=2pt, blue!80!black, opacity=0.3 ] (0,4) -- (4,4) -- (4,0);

  \node[left] at (0,0) {\small $c$};
  \node[left] at (0,2) {\small $p$};
  \node[left] at (0,4) {\small $c+1$};
  \node[below] at (0,0) {\small $c$};
  \node[below] at (2,0) {\small $p$};
  \node[below] at (4,0) {\small $c+1$};

  \node[above] at (0.5,4.05) {\bfseries\itshape \textcolor{blue!80!black}{B}};   
  \node[below] at (0.5,1.95) {\textit{\textbf{\textcolor{orange!80!black}{B\textsubscript{Z}}}}};  
  \node at (0.5,3) {\bfseries\itshape \textcolor{green!25!black}{W}};  

  \node at (2.8,2.5) {\bfseries\itshape push(U)};

  \node[right] at (4,3.8) {\bfseries\itshape U};

\end{tikzpicture}
\caption{Shape of $push(U)$ when $n=2$.}
\label{fig:pushU}
\end{figure}
The above lemma implies that \( \text{push}(U) \) is illustrated in Figure \ref{fig:pushU}, and 
\begin{gather*}
    \gamma_1(U) = \mu_{-} \qty( \qty{ y \in W \cup B_Z \ \Big| \ \text{there exists } t \in \mathbb{R} \text{ such that } x + t \mathbf{1} \in U } )
\end{gather*}
for any measurable upper set \( U \subseteq B \). We will fully leverage this equivalent characterization of $\gamma_1$ to prove sufficiency of the following Proposition \ref{prop:pushed measure} in Appendix \ref{proof:sufficiency of thm pushed measure}.

\begin{prop}\label{prop:pushed measure}
Suppose that the value of each item is uniformly distributed over $[c, c+1]$ identically and independently.
Then
$$\text{The optimal mechanism is uniform pricing }  \iff \gamma_1\succeq_{cvx}\mu_{+}.$$

\end{prop}

Indeed, the threshold \( c^* \) corresponds directly to the condition stated above, which asserts that the pushed measure convexly dominates the original measure, \( \mu_{+} \). 

\begin{prop}\label{prop:cvx iff c^*}
If the value of each item is uniformly distributed over \( [c, c+1] \) identically and independently,
\[
\gamma_1 \succeq_\text{cvx} \mu_+ \iff c \in [0, c^*].
\]
\end{prop}

\begin{proof}[Proof Sketch]
  The main reason for this relationship stems from the density functions of the respective measures. In Lemma \ref{lemma:density} in the Appendix \ref{proof:sufficiency of thm pushed measure}, we prove that  
\begin{eqnarray}\label{eq:density}
f^{\gamma_1}(x) \equiv
  \begin{cases}
    (n+1)(\min(x)-c)+c  & \text{ if }  2c+1-p > \max(x), \\
    (n+1)(c+1-p)        & \text{ if } \max(x) \geq 2c+1-p
  \end{cases}
\end{eqnarray}
can be interpreted as the density function of \( \gamma_1 \) in a certain sense. The density of \( \mu_{+} \) is calculated as \( c+1 \), and the condition \( (n+1)(c+1-p) \geq c+1 \) plays a crucial role in determining whether \( \gamma_1 \) convexly dominates \( \mu_{+} \). This condition is equivalent to \( c \leq c^* \), establishing the threshold for uniform pricing's optimality.

It is straightforward to see that if \( (n+1)(c+1-p) < c+1 \), then \( \gamma_1 \) does not convexly dominate \( \mu_{+} \). This follows from constructing a function that is positive only in the region where \( \max(x) \geq 2c+1-p \). The sufficiency argument is more intricate and will be established via induction on \( n \) in Appendix \ref{proof:uniform pricing is optimal if gamma_1>mu_+} and \ref{proof:sufficiency of thm pushed measure}.
\end{proof}

\subsection{Proof Sketch of Proposition \ref{prop:pushed measure}}\label{subsec:Proof Sketch}
This subsection explains the proof sketch of Proposition \ref{prop:pushed measure}. First, we will consider sufficiency of Proposition \ref{prop:pushed measure} by the following proposition:

\begin{prop}\label{prop:uniform pricing is optimal if gamma_1>mu_+}
uniform pricing is optimal if
\begin{gather*}
    \gamma_1\succeq_\text{cvx} \mu_{+}.
\end{gather*}
\end{prop}
\begin{proof}[Proof Sketch]
    Since \( \gamma_1 \) convexly dominates \( \mu_{+} \), the optimality condition is easily satisfied by choosing a measure \( \gamma^* \) such that \( \gamma^*_1 = \gamma_1 \) and \( \gamma^*_2 = \mu_{+} \). Then, the third condition follows immediately as  
\[
\gamma^*_1 - \gamma^*_2 \succeq_\text{cvx} \mu =\mu_{+} - \mu_{-}.
\]  
The remaining conditions are also verified using the properties of the pushed measure. For instance, the pushed measure ensures the almost sure (a.s.) condition by definition. Furthermore, the integral condition for \( u \) is also satisfied because \( \mu_{-} \) remains unchanged. Therefore, for the function \( U \), any modifications occur on a set of measure zero, which does not affect the integral.
\end{proof}
Next, we will explain the necessity part in Proposition \ref{prop:pushed measure}. As in Theorem 2 in \textcite{daskalakis2017strong}, when considering feasible dual variable $\delta\in\Gamma_{+}(X\times X)$, note that we can assume $\delta_1(X)=\delta_2(X)=\mu_+(X)$, $\delta_1\succeq_{cvx} \mu_{+}$ and $\mu_{-}\succeq_{cvx} \delta_2$ hold.
\begin{prop}\label{prop: nescessity} 
If  $\gamma_1\succeq_{cvx}\mu_{+}$ does not hold, there does not exist $\delta\in \Gamma_{+}(X\times X)$ such that
\begin{gather}
    \mu_+(X)=\delta_1(X)=\delta_2(X),\label{eq:delta moreover cdn}\\
    \delta_1\succeq_{cvx} \mu_{+},\label{eq:delta_1 cvx cdn}\\
    \mu_{-}\succeq_{cvx} \delta_2,\label{eq:delta_2 cvx cdn}\\ 
    u^*(x)-u^*(y)=\|x-y\|_{\infty}\quad\delta(x,y)\text{-almost surely},\label{eq:delta a.s. cdn}\\
    \int_{X}u^*d(\delta_1-\delta_2)=\int_{X}u^*d\mu.\label{eq:delta integral cdn}
\end{gather}
\end{prop}
Here, we provide the roadmap of the proof. We derive a contradiction by assuming that \( \gamma_1\succ_{cvx} \mu_{+} \) does not hold and that \( \delta\in \Gamma_{+}(X\times X) \) satisfies all conditions \eqref{eq:delta moreover cdn}, \eqref{eq:delta_1 cvx cdn}, \eqref{eq:delta_2 cvx cdn}, \eqref{eq:delta a.s. cdn}, and \eqref{eq:delta integral cdn}. In particular, we aim to show that \eqref{eq:delta_1 cvx cdn} fails under these assumptions. To achieve this, we establish the existence of a function \( u\in\mathcal{U}(X) \) such that  
$
\int_X u d\delta_1 < \int_X u d\mu_{+}.$ 

A natural candidate for \( u\in\mathcal{U}(X) \) arises from the fact that \( \gamma_1 \succeq_{cvx} \mu_{+} \) does not hold. This implies the existence of some function \( u\in\mathcal{U}(X) \) satisfying  
$
\int_X u d\gamma_1 < \int_X u d\mu_{+}.
$
By Proposition \ref{prop:cvx iff c^*}, we know that \( c < c^* \), and the density function in equation \eqref{eq:density} indicates that any function that is positive only in the region where \( \max(x) \geq 2x + 1 - p \) satisfies this inequality.

To be concrete, consider the function  
\[
v(x) \equiv \max\qty(\sum_{i=1}^n x_i - (nc+n+c-p), 0).
\]  
Since  
\[
v(c+1, \dots, c+1, 2c+1-p) = 0,
\]  
it follows that \( v(x) < 0 \) whenever \( \max(x) < 2c + 1 - p \), leading to  
\begin{gather}\label{eq:c>c^star implies gamma1 does not cvx dominate mu+}
    \int_X v d\gamma_1 < \int_X v d\mu_{+} 
\end{gather}

Our goal is now to show that the same inequality holds for any \( \delta \) satisfying all conditions. To proceed, we rely on three key assumptions: convex order \eqref{eq:delta_1 cvx cdn}, \eqref{eq:delta_2 cvx cdn}, the almost sure condition \eqref{eq:delta a.s. cdn}, and the integral equality condition \eqref{eq:delta integral cdn}.  

The main challenge is effectively using the almost sure condition and the integral equality condition. First, the convex order conditions (10) and (11) allow us to rewrite the integral equality condition (13) in a more useful form, which we will utilize in the following steps.  

\begin{rmk}\label{lemma:both delta integral equal}
Suppose that there exists $\delta\in \Gamma_{+}(X\times X)$ that satisfies  \eqref{eq:delta_1 cvx cdn}, \eqref{eq:delta_2 cvx cdn} and \eqref{eq:delta integral cdn}.
Then 
\begin{align*}
    \int_{X}u^*d\delta_1=\int_{X}u^*d\mu_+,\\
    \int_{X}u^*d\delta_2=\int_{X}u^*d\mu_-.
\end{align*}
\end{rmk}
\begin{proof}
Since $u^*$ is continuous, non-decreasing, and convex, \eqref{eq:delta_1 cvx cdn} and \eqref{eq:delta_2 cvx cdn} mean
\begin{gather*}
    \int_{X}u^*d\delta_1\geq \int_{X}u^*d\mu_+,\\
    \int_{X}u^*d\delta_2\leq \int_{X}u^*d\mu_-.
\end{gather*}
Here, \eqref{eq:delta integral cdn} implies that both inequalities hold with equality.

\end{proof}

Additionally, (10) and (11) also allow us to derive a fundamental property about \( \delta \), which will serve as the foundation for the subsequent discussion. 

\begin{rmk}\label{lemma:divide Z and W}
Suppose that there exists $\delta\in \Gamma_{+}(X\times X)$ that satisfies \eqref{eq:delta moreover cdn}, \eqref{eq:delta_1 cvx cdn}, \eqref{eq:delta_2 cvx cdn} and, \eqref{eq:delta integral cdn}.
\begin{gather*}
    \delta_1(Z)=\mu_+(Z)\\
    \delta_1\left(B\right)=\mu_+\left(B\right)
\end{gather*}
In particular, the second equation and \eqref{eq:delta moreover cdn} implies that $\delta_1$ that satisfies \eqref{eq:delta moreover cdn}, \eqref{eq:delta_1 cvx cdn}, \eqref{eq:delta_2 cvx cdn} and, \eqref{eq:delta integral cdn} has measure only on $Z\cup B$.
\end{rmk}
\begin{proof}The proof is in Appendix \ref{proof:necessity of thm pushed measure}. \end{proof}


Next, the following lemma provides an effective way to utilize the almost-sure condition. This lemma illustrates how the almost-sure condition constrains the measure. For example, suppose that \( \delta_2 = \mu_{-} \). Then, this implies that \( \delta_1 \) is always weakly smaller than the pushed measure \( \gamma \). Overall, this lemma helps establish an upper bound for \( \delta_1 \), which is key to demonstrating why the inequality holds for any \( \delta \).

\begin{rmk}\label{lemma:pushu}
Suppose that there exists $\delta\in \Gamma_{+}(X\times X)$ that satisfies \eqref{eq:delta a.s. cdn}.
Then the following two equations hold:
\begin{gather*}
    \delta_1(Z\setminus B_Z)=\delta_2(Z\setminus B_Z)\\
    \delta_1(U\cap W)\leq \delta_2(push(U)\cap W)+(\delta_2-\delta_1)(push(U)\cap B_Z) \text{ for  all upper set $U\subset B$ with respect to $B$}.
\end{gather*}
\end{rmk}

\begin{proof}The proof is in Appendix \ref{proof:necessity of thm pushed measure}.
\end{proof}

Moreover, the Lemma \ref{lemma:pushu} implies the following lemma. First, we will define the pushed function. 

\begin{dfn}\label{dfn:push of function}
    For a bounded function $b:B\rightarrow \mathbb{R}$, $b_{push}:X\rightarrow \mathbb{R}$ is defined as 
    \begin{gather*}
        b_{push}(x)=b(x+t\bm{1}), \text{ where } t=c+1-\max(x).
    \end{gather*}
\end{dfn}

\begin{rmk}\label{lemma:integral inequality between delta_1 and delta_2}
Suppose that $c>c^* $ and there exists $\delta\in \Gamma_{+}(X\times X)$ that satisfies \eqref{eq:delta moreover cdn}, \eqref{eq:delta_1 cvx cdn}, \eqref{eq:delta_2 cvx cdn}, \eqref{eq:delta a.s. cdn} and \eqref{eq:delta integral cdn}.
Then, for all nonnegative, bounded and nondecreasing functions $b:B \rightarrow [0, \infty)$,
\begin{gather}
\int_W b(x)d\gamma_1=\int_W b_\text{push}(x)d\mu_-,\label{eq:gamma push equal}\\
\int_W b(x)d\delta_1\leq \int_W b_\text{push}(x)d\delta_2+\int_{B_Z}b_\text{push}(x)d\qty(\delta_2-\delta_1),\label{eq:delta push inequal}
\end{gather}
\end{rmk}
\begin{proof}
    The proof is in Appendix \ref{proof:necessity of thm pushed measure}.
\end{proof}

The following lemma completes the proof of Proposition \ref{prop: nescessity} by utilizing the integral equality condition. By using the integral equality condition, we can freely add the $u^*$ to the original function, making it convex. 

\begin{rmk}\label{lemma:final lemma for necessity}
Suppose that $c>c^* $ and there exists $\delta\in \Gamma_{+}(X\times X)$ that satisfies \eqref{eq:delta moreover cdn}, \eqref{eq:delta_1 cvx cdn}, \eqref{eq:delta_2 cvx cdn}, \eqref{eq:delta a.s. cdn} and \eqref{eq:delta integral cdn}.
Then $\delta_1\succeq_{cvx}\mu_+$ does not hold.
\end{rmk}
\begin{proof}[Proof Sketch]
Let $\Tilde{v}:B\to[0, \infty)$ be the restriction of $v$ to $B$, i.e. $\Tilde{v}\equiv v|_B$. Define $v_t:X\rightarrow[0, \infty)$ as
\begin{gather*}
    v_t(x)\equiv\max\qty(\Tilde{v}_{push}(x)+(n+t)(\max(x)-p),0)
\end{gather*}
for $t>0$.
We can prove that 
\begin{gather*}
\Tilde{v}_{push}(x)+(n+t)(\max(x)-p)
    =\max\qty(\sum_{i=1}^n x_i-c-(n-1)p,n(\max(x)-p))+t(\max(x)-p).
\end{gather*}
Since the pointwise maximum of affine functions is convex and all affine terms are non-decreasing, their sum remains convex. Thus, $v_t$ is in $\mathcal{U}(X)$ and by \eqref{eq:delta_2 cvx cdn},
\[
\int_X v_t d\qty(\mu_--\delta_2)\geq 0 \text{  for all $t>0$. }
\]
By $v_t(x)=\Tilde{v}_{push}(x)+(n+t)u^*(x)$ for $x\in W\cup B_Z$ and  Lemma \ref{lemma:both delta integral equal},
\begin{eqnarray*}
    0
    \leq \int_X v_t d\qty(\mu_--\delta_2)
    &=& \int_{W\cup B_Z} v_t d\qty(\mu_--\delta_2) + \int_{Z\setminus  B_Z} v_t d\qty(\mu_--\delta_2)\\
    &=&\int_{W\cup B_Z} \Tilde{v}_{push} d\qty(\mu_--\delta_2) + \int_{Z\setminus  B_Z} v_t d\qty(\mu_--\delta_2) \quad(\because \text{ Lemma }\ref{lemma:both delta integral equal})
\end{eqnarray*}
By the definition of $v_t(x),$ notice that for all $x\in Z\setminus B_Z$, $v_t(x)\rightarrow 0$ as $t\rightarrow \infty$. Thus, by the bounded convergence theorem,
 \begin{gather*}
     0\leq \int_{W\cup B_Z} \Tilde{v}_{push} d\qty(\mu_--\delta_2)+\lim_{t\rightarrow \infty}\int_{Z\setminus B_Z} v_t d\qty(\mu_--\delta_2) = \int_{W\cup B_Z} \Tilde{v}_{push} d\qty(\mu_--\delta_2).
 \end{gather*}
 
Moreover, by \eqref{eq:gamma push equal} and \eqref{eq:delta push inequal}, we obtain
    \begin{eqnarray*}
        0&\le&\int_{W\cup B_Z} \Tilde{v}_{push} d\qty(\mu_--\delta_2)
        \\&\le& \int_{W} \Tilde{v}_{push} d\qty(\mu_--\delta_2)+\int_{B_Z} \Tilde{v}_{push} d\qty(\delta_1-\delta_2)\quad (\because \text{ By calculation})\\
        &=& \int_{W} \Tilde{v} d\qty(\gamma_1-\delta_1) \quad(\because \eqref{eq:gamma push equal}\eqref{eq:delta push inequal})
        \\&=& \int_{X} v d\qty(\gamma_1-\delta_1) \quad (\because \text{ By nature of $\gamma,\delta$, and $v$})
\end{eqnarray*}
This and \eqref{eq:c>c^star implies gamma1 does not cvx dominate mu+} completes the proof.
\end{proof}

\section{The Automated Mechanism Design Approach}\label{sec:AMD}
Our result was conjectured by the Automated Mechanism Design approach, which leverages computer algorithms to identify optimal mechanisms. We first numerically derive a candidate for the optimal mechanism in small-scale problems with specific distributions using the deep-learning-based algorithm developed by \textcite{dutting2023optimal}. We then generalize this mechanism to $n$-item setting with a broader class of distributions and employ the duality approach to establish its optimality formally. In addition to identifying the optimal mechanism, we use computational methods to approximate the optimal dual variables. In hindsight, our simple optimal mechanism might have been discovered without computational assistance; however, deriving it from first principles alone would have been nearly infeasible. The Automated Mechanism Design approach thus highlights a promising role for AI in theoretical economic research—namely, automating part of the research process by using computer algorithms to infer optimal mechanisms.

\subsection{The primal problem \eqref{problem:primal}}
For the primal problem \eqref{problem:primal}, \textcite{dutting2023optimal} provides useful methods when the number of items is small. We utilize the \textit{RochetNet} architecture introduced in their paper. For our analysis, we consider the following three cases:

\begin{enumerate}
    \item In the first case, we assume that the number of items is $n=2$, and the valuations of item 1 ($x_1$) and item 2 ($x_2$) are independently and identically distributed (i.i.d.) according to the uniform distribution $U[0,1]$ (Figure \ref{fig:prime_U01}).
    
    \item In the second case, we again consider $n=2$, where $x_1$ follows the cumulative distribution function (CDF) $F_1(x_1) = x_1^{1/2}$, and $x_2$ follows the CDF $F_2(x_2) = x_2^2$, independently of $x_1$ (Figure \ref{fig:beta}).
    
    \item In the third case, we extend the analysis to $n=3$, assuming that $x_1$, $x_2$, and $x_3$ are i.i.d. according to the uniform distribution $U[1, 2]$ (Figure \ref{fig:U12}).
\end{enumerate}

The three figures made by the RochetNet architecture illustrate the allocation probability of each item. Dark red represents a probability of 1, while yellow represents a probability of 0. The transition from dark red to yellow occurs when the maximum valuation of the items exceeds a certain price threshold.
Based on the figures, uniform pricing appears to be optimal for all cases.

\begin{figure}[htbp]
\centering
\includegraphics[width=12cm]{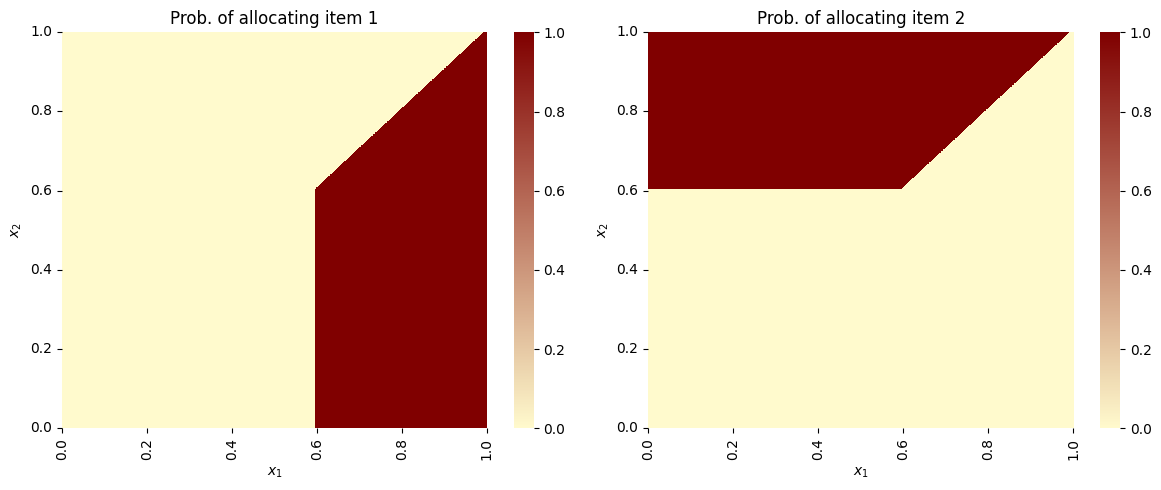}
\caption{The optimal allocation function in the case where the valuations of two items are independently drawn from $U[0,1]$.}
\label{fig:prime_U01}
\end{figure}

\begin{figure}[htbp]
\centering
\includegraphics[width=12cm]{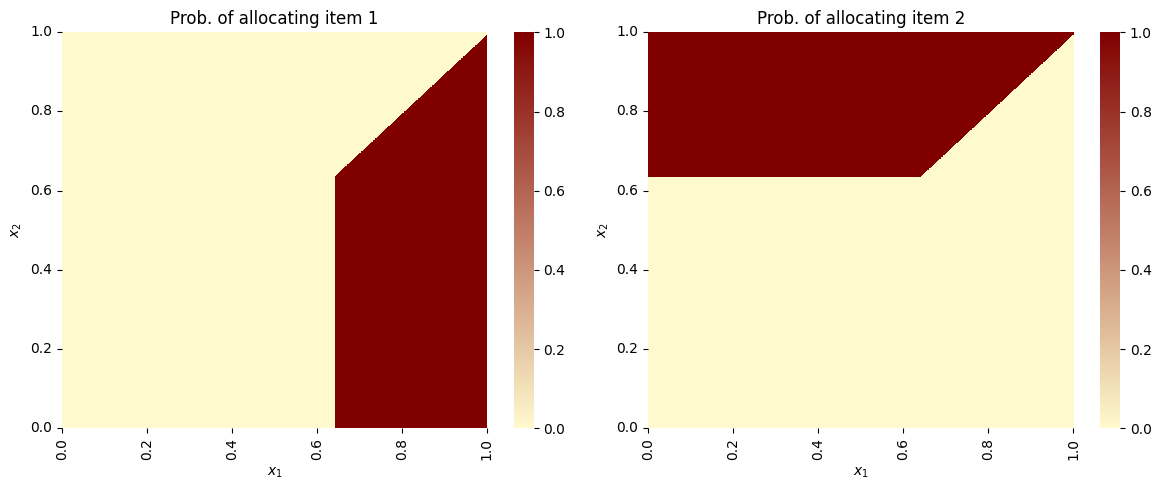}
\caption{The optimal allocation function in the case where the valuations of two items are independently drawn from $x_1\sim F_1$ and $x_2\sim F_2$.}
\label{fig:beta}
\end{figure}

\begin{figure}[htbp]
\centering
\includegraphics[width=12cm]{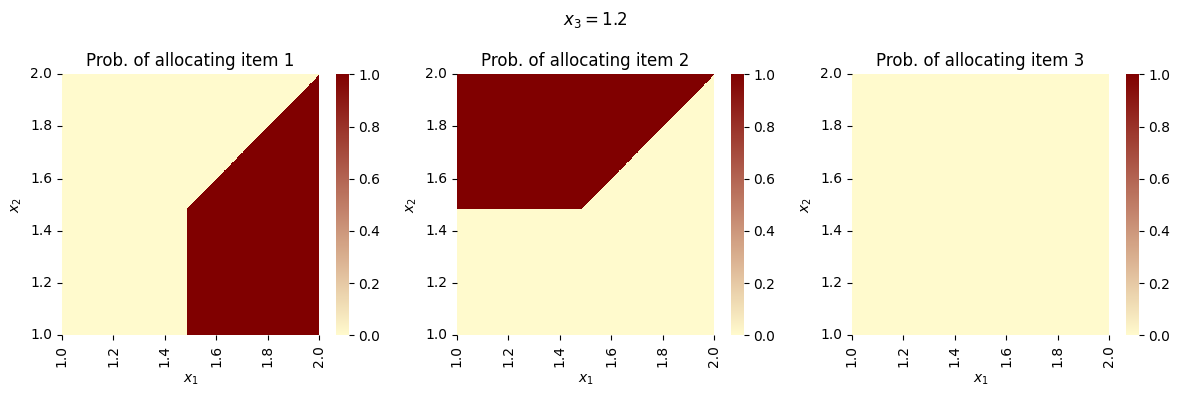}
\includegraphics[width=12cm]{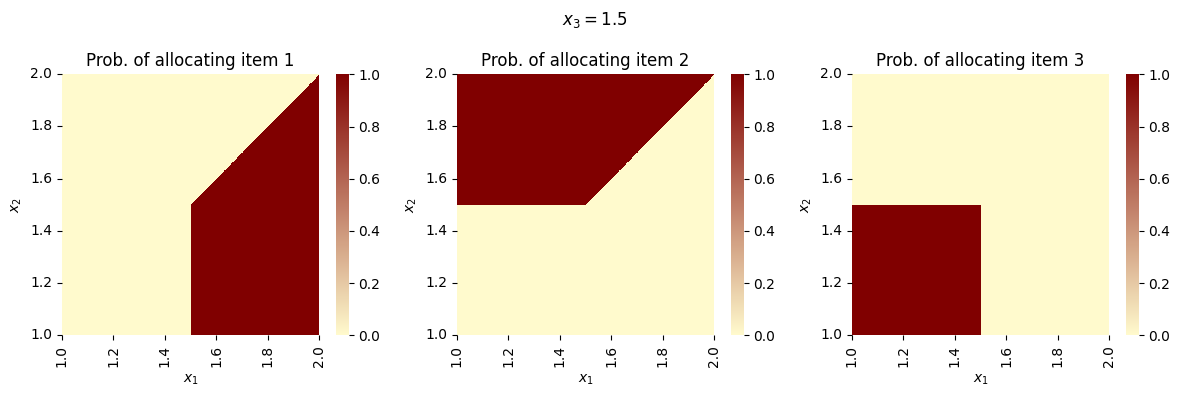}
\includegraphics[width=12cm]{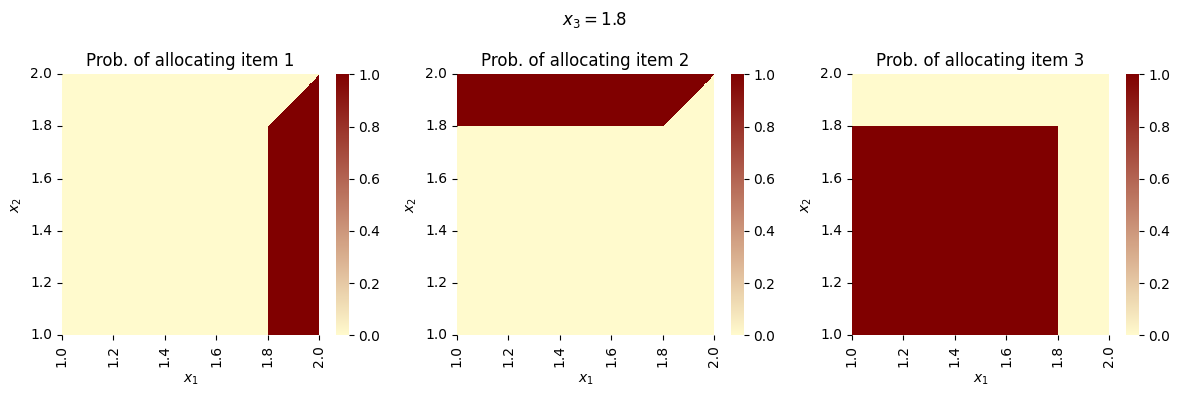}
\caption{The optimal allocation function in the case where the valuations of three items are independently drawn from $U[1,2]$. The results are plotted by fixing $x_3$ at 1.2, 1.5, and 1.8 while varying $x_1$ and $x_2$.}
\label{fig:U12}
\end{figure}

\section{Conclusion}
For over 40 years,  the optimal mechanism for multiple items has been an unsolved problem mainly due to its difficulty in proving the optimality. Here, we provide broadly applicable and easy-to-interpret results, focusing on a single buyer who can consume one item. 

First, we show that selling each item at an identical price is optimal for any number of items when the value distributions satisfy specific condition called scale monotonicity. Many famous distribution satisfies this condition such as beta distributions. This condition is in fact equivalent to sufficient condition by \textcite{haghpanah2021pure} in symmetric case.

Similarly, we demonstrate that bundling all items together is also optimal when and only when the values of multiple goods are independently drawn from $U[c,c+1]$, subject to specific conditions on $c$ that determine the feasibility of bundling. To verify the sufficiency and necessity, we defined a new notion of pushed set and pushed measure and applied them to the strong duality by \textcite{daskalakis2017strong}.

To identify the optimal mechanism, we employed the Automated Mechanism Design (AMD) approach, using the deep learning algorithm by \textcite{dutting2023optimal} to numerically solve small-sized problems for specific distributions. As for the second result, these numerical insights were generalized through an educated conjecture and rigorously verified via duality theory. This combination of AMD and duality provides a practical and efficient framework for verifying optimal solutions, making mechanism design more accessible for real-world applications in policy and resource allocation.

\clearpage

\section{Appendix}

\subsection{Omitted proof in Subsection \ref{sec:when's alpha}}
In this subsection, we prove Theorem \ref{thm:when's alpha}. To do that, first we establish some lemma.

Consider a pair of jointly distributed random variables $(x, y) \in \mathbb{R}^S \times \mathbb{R}$, where $S$ is a finite set.  
For each $s \in S$, define the conditional cumulative distribution function of $x(s)$ given $y$ by  
\[
F_s(z \mid y) := \Pr[x(s) \leq z \mid y].
\]  
Let $F_s^{-}(\cdot \mid y)$ denote the left limit of $F_s(\cdot \mid y)$, and let $F_s^{-1}(\cdot \mid y)$ denote its generalized inverse, i.e.
\begin{align*}
    F_s^{-}(z \mid y)
    &\equiv\lim_{h\to-0}F_s^{}(z+h \mid y)\\
    F^{-1}_{s}(\hat{q}\mid y)
    &\equiv\inf\qty{z\in \mathbb{R}|F_{s}(z|y)\geq \hat{q}}
\end{align*}

\begin{dfn}
Let $r \in \mathbb{R}$ be a random variable uniformly distributed on $[0,1]$, independent of $(x, y)\in \mathbb{R}^S \times \mathbb{R}$.  
Define a random variable $q \in \mathbb{R}^S$ by  
\[
q(s) := r F_{s}(x(s) \mid y) + (1 - r) F_{s}^{-}(x(s) \mid y).
\]  
Then $q$ is called a quantile of $x$ conditional on $y$.
\end{dfn}

\begin{rmk}[based on \textcite{ruschendorf2009distributional}]\label{lem: property of quantile}
Let $q \in \mathbb{R}^S$ be a quantile of $x \in \mathbb{R}^S$ conditional on $y \in \mathbb{R}$.  
Then, conditional on $q = \hat{q}$, we have  
\[
x(s) = F_{s}^{-1}(\hat{q}(s) \mid y) \quad \text{for all } s \in S,
\]  
with probability one.
\end{rmk}

\begin{proof}
Fix any $s \in S$.  
Again by \textcite{ruschendorf2009distributional}, conditional on $y = \hat{y}$, we have  
\[
x(s) = F_s^{-1}(q(s) \mid \hat{y}) \quad \text{almost surely}.
\]  
Therefore,  
\[
x(s) = F_s^{-1}(q(s) \mid y) \quad \text{almost surely}.
\]  
It follows that, conditional on $q = \hat{q}$,  
\[
x(s) = F_s^{-1}(\hat{q}(s) \mid y) \quad \text{almost surely}.
\]
\end{proof}

\begin{rmk}[based on \textcite{strassen1965existence,kamae1977stochastic}]\label{lem: strassen quantile}
Let $q \in \mathbb{R}^S$ be a quantile of $x \in \mathbb{R}^S$ conditional on $y \in \mathbb{R}$. Assume $q$ and $y$ are independently distributed.
The distribution of $x$ is stochastically non-decreasing in $y$ if and only if  
$F_s^{-1}(\hat{q}(s) \mid y)$
is non-decreasing in $y$ for all $\hat{q}$ and all $s \in S$.
\end{rmk}
\begin{proof}
The “if” direction follows from \textcite{strassen1965existence, kamae1977stochastic}.
In what follows, we prove the “only if” direction.
Fix any $\hat{q}$ and $s \in S$.
Suppose that the distribution of $x$ is stochastically non-decreasing in $y$,
or equivalently $\Pr[x \in U \mid y = \hat{y}]$ is non-decreasing in $\hat{y}$ for all upper sets $U \subseteq \mathbb{R}^S$.
Take any $\hat{y}_1,\hat{y}_2\in \mathbb{R}$ such that $\hat{y}_1<\hat{y}_2$.
Let
\[U=\left\{\hat{x}\in \mathbb{R}^S\mid \hat{x}(s)\geq F_s^{-1}(\hat{q}(s) \mid \hat{y}_1)\right\}.\]
Note that $U$ is an upper set.
Thus,
\[
1-\hat{q}(s)\leq \Pr[x \in U \mid y = \hat{y}_1]\leq \Pr[x \in U \mid y = \hat{y}_2].
\]
It follows that
\[
F_s^{-1}(\hat{q}(s) \mid \hat{y}_1)\leq F_s^{-1}(\hat{q}(s) \mid \hat{y}_2).
\]
\end{proof}

Here, e recall some important notations.

Let $F : [\underline{v}, \overline{v}] \to [0, 1]$ be a CDF that has a positive density $f$ over $[\underline{v}, \overline{v}]$ (so that $F$ is continuous and strictly increasing on $[\underline{v}, \overline{v}]$. ) 
Denote $T=[\underline{v}, \overline{v}]^n$ as a type space of the buyer and $t_i$ as a buyer's value of item $i$ for each $i=1, \ldots, n$.\par

Suppose $n\geq 2$ and $v(b,t)$ is unit-demand, i.e., $v(b,t)=\max_{i\in b} t_i$. We assume $v(\emptyset, t)=0$ for all $t\in T$. Let $r(\cdot,t) \in \mathbb{R}^B$ and $v(b^*,t) \in \mathbb{R}$ be the profile of relative values of type $t$ and the value of the grand bundle of type $t$ respectively, where $B=2^\qty{1,\ldots,n}$.

In what follows, we apply above results (Lemma\ref{lem: property of quantile} and Lemma\ref{lem: strassen quantile}) to the case where $x$ represents the profile of relative values $r(\cdot,t) \in \mathbb{R}^B$, $y$ represents the value of the grand bundle $v(b^*,t) \in \mathbb{R}$ and $S=B=2^\qty{1,\ldots,n}$. In this special setting, the following hold.

\begin{rmk}\label{lem:q and v are independent}
    For each $i=1, \ldots, n$, take any $\alpha_i>0$ and  consider the CDF $F^{\alpha_i}$ over $[\underline{v}, \overline{v}]$. Suppose $t_i$ follows CDF $ F^{\alpha_i}$ independently.
    Let $q \in \mathbb{R}^B$ be a quantile of $r(\cdot,t) \in \mathbb{R}^B$ conditional on $v(b^*,t) =\hat{v}$. Then, $q$ and $\hat{v}$ are independent.  
\end{rmk}

\begin{proof}
    By \textcite{ruschendorf2009distributional}, conditional on $v(b^*,t) =\hat{v}$, the random variable $q(b)$ is uniformly distributed on $[0,1]$ for each $b\in B$. We need to show $q=(q(1), \ldots, q(b^*))$ and $\hat{v}$ are independent. First, notice that conditional on $v(b^*,t) =\hat{v}$, the probability that item $i$ has the highest value among others is  $\frac{\alpha_i}{\sum_{j=1}^n\alpha_j}$. Thus, conditional on $v(b^*,t) =\hat{v}$, the CDF graph of $r(\{i\}, t)$ is as in Figure\ref{fig:cdf of r_i}.

\begin{figure}[htbp]
\centering
\begin{tikzpicture}[>=stealth, scale=5]

\def\xa{0}          
\def\xb{1.05}       
\def\ya{0}          
\def\yb{1.05}       
\def\th{\linewidth} 
\def\yalpha{0.58}   

\draw[->, line width=0.6pt] (\xa,\ya) -- (\xb,\ya) node[below right=2pt] {$r(\{i\}, t)$};
\draw[->, line width=0.6pt] (\xa,\ya) -- (\xa,\yb);

\node[below] at (\xa,\ya) { $\frac{\underline{v}}{\overline{v}}$};
\node[below] at (1,\ya) { 1};
\node[left]  at (\xa,1)   {\small 1};

\node[above] at (0.1,1.05) {\small CDF of $r(\{i\}, t)$ };

\draw[dashed] (1,\ya) -- (1,1.02);
\draw[dashed] (\xa+0.02,1) -- (1,1);
\draw[dashed] (\xa+0.02,\yalpha) -- (1,\yalpha);

\node[left] at (\xa,\yalpha)
    {$\displaystyle \frac{\sum_{j\not=i}\alpha_j}{\sum_{j=1}^{n}\alpha_j}$};

\draw[line width=0.7pt]
  (0,0) .. controls (0.35,0.25) and (0.7,0.45) .. (1,\yalpha);

\draw[line width=0.6pt] (1,\yalpha) circle[radius=0.018];

\fill (1,1) circle[radius=0.02];

\end{tikzpicture}
\caption{CDF of $r(\{i\}, t)$, conditional on $v(b^*,t) =\hat{v}$ (We don't care the convexity of the graph. It might not be concave as in graph)}
\label{fig:cdf of r_i}
\end{figure}

Also, conditional on $v(b^*,t) =\hat{v}$, there exists item$i$ such that $r(\{i\}, t)=1$. For such item$i$, Figure\ref{fig:cdf of r_i} implies that 
\begin{gather}\label{eq:q(i)}
    q(\{i\})\ge\frac{\sum_{j\not=i}\alpha_j}{\sum_{j=1}^{n}\alpha_j}
\end{gather} 
For other item$j\not=i$, notice that $q(\{j\})$ is uniformly distributed on $[0, \frac{\sum_{k\not=j}\alpha_k}{\sum_{k=1}^{n}\alpha_k}]$. Moreover, given that \eqref{eq:q(i)}, $q(\{1\}), \ldots, q(\{i-1\}), q(\{i+1\}), \ldots, q(\{n\})$ are all independent since $t_1\ldots, t_n$ are all independently distributed. 
Here, notice that, we could characterize the distribution of $q$  independent of $\hat{v}$. In Figure\ref{fig:distribution of q}, the distribution of $(q(\{1\}, q(\{2\})))$ is shown.

\begin{figure}[htbp]
\centering
\begin{tikzpicture}[>=stealth, scale=4]

\def\xa{0} \def\xb{1}
\def\ya{0} \def\yb{1}
\def\xaone{0.55} 
\def\yatwo{0.6}  

\draw[->] (\xa,\ya) -- (1.15,0) node[below] {$q(\{1\})$};
\draw[->] (\xa,\ya) -- (0,1.15) node[left] {$q(\{2\})$};

\draw[thick] (\xa,\ya) rectangle (\xb,\yb);

\draw[dashed] (\xaone,\ya) -- (\xaone,\yb);
\draw[dashed] (\xa,\yatwo) -- (\xb,\yatwo);

\node[below] at (0,0) {0};
\node[left]  at (0,1) {1};
\node[below] at (\xaone,0) {$\tfrac{\alpha_2}{\alpha_1+\alpha_2}$};
\node[left]  at (0,\yatwo) {$\tfrac{\alpha_1}{\alpha_1+\alpha_2}$};
\node[below] at (1,0) {1};

\fill[pattern=north east lines] (\xa,\yatwo) rectangle (\xaone,\yb); 
\fill[pattern=north east lines] (\xaone,\ya) rectangle (\xb,\yatwo); 

\node at (0.3,0.8) {A};
\node at (0.8,0.3) {B};

\end{tikzpicture}
\caption{Distribution of $(q(\{1\}, q(\{2\}))$ when $n=2$. A probability mass of $\frac{\alpha_2}{\alpha_1+\alpha_2}$ is uniformly distributed in rectangle A and a probability mass of $\frac{\alpha_1}{\alpha_1+\alpha_2}$ is uniformly distributed in rectangle B.}
\label{fig:distribution of q}
\end{figure}

Futhermore, we can specify joint distribution of $(q(\{1\}), \ldots, q(\{b^*\}))$ only by the joint distribution of $(q(\{1\}), \ldots, q(\{n\}))$. These imply that $(q(\{1\}), \ldots, q(\{n\}))$ and $\hat{v}$ are independent. 
\end{proof}
By Lemma\ref{lem:q and v are independent}, the results of Lemma\ref{lem: property of quantile} and Lemma\ref{lem: strassen quantile} hold in this setting.

Recall the definition of scale monotonicity.
\begin{dfn*}
A CDF $F:[\underline{v}, \overline{v}]\to [0,1]$ is scale-monotone if,
for all $\omega \in (\underline{v}/\overline{v},1)$, the function
\[
x \mapsto \frac{F(\omega x)}{F(x)}
\]
is non-increasing in $x \in (\underline{v}/w, \overline{v}]$.
\end{dfn*}

Now we present Theorem \ref{thm:when's alpha} again.

\begin{thm*}
    Assume $F$ is scale-monotone. For each $i=1, \ldots, n$, take any $\alpha_i>0$ and  consider the CDF $F^{\alpha_i}$ over $[\underline{v}, \overline{v}]$. Suppose $t_i$ follows CDF $ F^{\alpha_i}$ independently. Then for any number of items and for any $\alpha_1, \ldots, \alpha_n>0$, uniform pricing is optimal.
\end{thm*}

To prove Theorem \ref{thm:when's alpha}, we define quantile-scaled monotonicity and show that it's equivalent to scale monotonicity.

\begin{dfn}
    A strictly increasing CDF $F : [\underline{v}, \overline{v}] \to [0, 1]$ is quantile-scaled monotone if, for 
all $q \in (0, 1)$, the value $\omega \in (0, 1)$ which satisfies
\[
\frac{F(\omega y)}{F(y)} = q
\]
is non-decreasing in $y$ for $y\in(\underline{v}, \overline{v}]$
\end{dfn}
Since $F$ is continuous and strictly increasing, note that such $\omega$ is uniquely determined for each $q \in (0, 1)$ and $y \in (\underline{v}, \overline{v}]$

\begin{rmk}\label{lem: scale monotone iff quantile scaled monotone}
A strictly increasing CDF $F$ is scale-monotone if and only if $F$ is quantile-scaled monotone.
\end{rmk}

\begin{proof}
Suppose that $F$ is quantile-scaled monotone.
Fix any $w\in(\underline{v}/\overline{v},1)$ and take any
$x_1,x_2\in (\underline{v}/w, \overline{v}]$ with $x_1<x_2$.
Define
\[
q:=\frac{F(wx_2)}{F(x_2)}.
\]
By quantile-scaled monotonicity, there exists $\hat{w}\in (0,w]$ such that
\[
\frac{F(\hat{w}x_1)}{F(x_1)}=q.
\]
Since $\hat{w}\le w$ and $F$ is increasing, it follows that
\[
\frac{F(wx_1)}{F(x_1)}
\geq \frac{F(\hat{w}x_1)}{F(x_1)}
= \frac{F(wx_2)}{F(x_2)}.
\]

Conversely, suppose that $F$ is scale-monotone.
Take any $q\in (0,1)$ and $x_1,x_2\in (\underline{v},\overline{v}]$ with $x_1<x_2$.
Let $w_1,w_2\in (0,1)$ satisfy
\[
\frac{F(w_1x_1)}{F(x_1)}
=\frac{F(w_2x_2)}{F(x_2)}
=q.
\]
By scale-monotonicity,
\[
\frac{F(w_2x_2)}{F(x_2)}
=\frac{F(w_1x_1)}{F(x_1)}
\geq \frac{F(w_1x_2)}{F(x_2)},
\]
which implies $w_1\leq w_2$.
\end{proof}

By Lemma \ref{lem: scale monotone iff quantile scaled monotone}, proving Theorem \ref{thm:when's alpha} is equivalent to prove the following proposition.

\begin{prop}\label{prop:quantile scaled monotone}
    Assume $F$ is quantile-scaled monotone. For each $i=1, \ldots, n$, take any $\alpha_i>0$ and  consider the CDF $F^{\alpha_i}$ over $[\underline{v}, \overline{v}]$. Suppose $t_i$ follows CDF $ F^{\alpha_i}$ independently. Then for any number of items and for any $\alpha_1, \ldots, \alpha_n>0$, uniform pricing is optimal.
\end{prop}

\begin{proof}
    We will prove that the profile of relative values is stochastically non-decreasing in the positive part of the grand value when the seller offers \( n \) items, where the value of the \( i \)-th item follows \( F^{\alpha_i} \). If this is true, \textcite{haghpanah2021pure} implies that uniform pricing is optimal in that case.
   \par
    Let $F_{b}(\omega|y)\equiv \text{Pr}_{t_i\sim F^{\alpha_i},\forall i}[r(b,t)\leq \omega | v(b^*,t)=y]$ for all $\omega\in [0,1]$ and $y\in  \mathbb{R}$.
    We will prove that $F^{-1}_{b}(\hat{q}|y)$\footnote{This inverse function is generalized one, i.e., $F^{-1}_{b}(\hat{q}\mid y)\equiv\inf\qty{z\in \mathbb{R}|F_{b}(z|y)\geq \hat{q}}.$}  is non-decreasing in $y>0$ for almost all $\hat{q}\in [0,1]$ and all $b\in B$. This and Lemma \ref{lem: strassen quantile}  complete the proof. Since this condition trivially holds for $b=\emptyset$ or $b=b^*$, we take $b\in B\backslash\qty{\emptyset, b^*}$ arbitrarily. 
    \par
    First, we will calculate $F_{b}(\omega|y)$ for each $\omega\in [0,1]$. Let $\bar{b}=B\backslash b$, $A=\sum_{i=1}^n\alpha_i$, $A_b=\sum_{i\in b}\alpha_i$, and $A_{\bar{b}}=\sum_{i\notin b}\alpha_i$. Since $b\in B\backslash\qty{\emptyset, b^*}$ and $n\geq 2$, $\bar{b}$ is nonempty and $A_b, A_{\bar{b}}>0$.
    
    \begin{enumerate}
        \item If $\omega\in[0, 1)$,
            \begin{eqnarray*}
                F_{b}(\omega|y)&=&\text{Pr}[r(b,t)\leq \omega | v(b^*,t)=y]
                \\&=&\dfrac{\text{Pr}[v(b,t)\leq \omega y \text{ and } v(b^*,t)=y]}{\text{Pr}[v(b^*,t)=y]}
                \\&=&\dfrac{\text{Pr}[v(b,t)\leq \omega y \text{ and } v(\bar{b},t)=y]}{\text{Pr}[v(b^*,t)=y]}(\because \omega<1 )\\
                &=&\dfrac{\text{Pr}[v(b,t)\leq \omega y] \cdot\text{Pr}[v(\bar{b},t)=y]}{\text{Pr}[v(b^*,t)=y]}\quad(\because b \text{ and $\bar{b}$ are independent})
                \\&=& \dfrac{F(\omega y)^{A_b} \cdot A_{\bar{b}}f(y)F(y)^{A_{\bar{b}}-1}}{Af(y)F(y)^{A-1}}\quad(\because  \text{ each item $i$ is independently distributed according to $F^{\alpha_i}$})
                \\&=&\dfrac{A_{\bar{b}}}{A}\cdot  \qty(\dfrac{F(\omega y)}{F(y)})^{A_b}
            \end{eqnarray*}
            where third equality holds since $\omega<1$ implies:
            $$[v(b,t)\leq \omega y \text{ and } v(b^*,t)=y] \iff [v(b,t)\leq \omega y \text{ and } v(\bar{b},t)=y].$$
            Since $y>0$ and $F(t)$ is continuous and strictly increasing on $[0, \bar{v}]$, $F_{b}(\omega|y)$ is continuous and strictly increasing with respect to $\omega$ if $\omega\in[0, 1)$.  Notice $\lim_{\omega\rightarrow 1}F_{b}(\omega|y) = \dfrac{A_{\bar{b}}}{A}$ holds. 
        \item If $\omega=1$, then $F_{b}(\omega|y)=1$.
    \end{enumerate}
Therefore, 
\begin{gather}   
F_{b}(\omega|y) = 
    \begin{cases}
        \dfrac{A_{\bar{b}}}{A}\cdot  \qty(\dfrac{F(\omega y)}{F(y)})^{A_b} & \text{if } \omega\in[0,1) \\
        1 & \text{if } \omega=1
    \end{cases}
\end{gather}

Second, we will calculate $F^{-1}_{b}(\hat{q}|y)$ for each $\hat{q}\in [0,1]$.
    \begin{enumerate}
        \item If $\hat{q}\in[0, \frac{A_{\bar{b}}}{A})$, since $F_{b}(\omega|y)$ is continuous and strictly increasing with respect to $\omega$ if $\omega\in[0, 1)$, $F^{-1}_{b}(\hat{q}|y)$ can be calculated by simply inverting it. Thus, $F^{-1}_{\alpha}(\hat{q}|y)$ is a value  $\omega\in[0, 1)$ which satisfies:
            \begin{gather}\label{eq: quantile}
                \dfrac{A_{\bar{b}}}{A}\cdot  \qty(\dfrac{F(\omega y)}{F(y)})^{A_b} = \hat{q}\Leftrightarrow  \dfrac{F(\omega y)}{F(y)}=\qty(\dfrac{A}{A_{\bar{b}}}\hat{q})^{\frac{1}{A_b}}.
            \end{gather}
        \item If $\hat{q}\in[\frac{A_{\bar{b}}}{A}, 1]$, then $F^{-1}_{b}(\hat{q}|y)=\inf\qty{z\in \mathbb{R}|F_{b}(z|y)\geq \hat{q}}=1$.
    \end{enumerate}
Therefore, 
\begin{gather} 
    F_{b}^{-1}(\hat{q}|y) = 
    \begin{cases}
        \omega \text{ which satisfies \eqref{eq: quantile}} & \text{if } \hat{q} \in\left[0,   \dfrac{A_{\bar{b}}}{A}\right) \\
        1 & \text{if } \hat{q} \in\left[  \dfrac{A_{\bar{b}}}{A}, 1\right] \\
    \end{cases}
\end{gather} 
Let us change the variable as $q^{\prime}\equiv\qty(\dfrac{A}{A_{\bar{b}}}\hat{q})^{\frac{1}{A_b}}$ and let $G(q^{\prime}|y)\equiv F_{b}^{-1}(\hat{q}|y)$. We obtain 
\begin{gather*} 
    G(q^{\prime}|y) = 
    \begin{cases}
        \omega \text{ which satisfies } \dfrac{F(\omega y)}{F(y)}=q^{\prime} & \text{if } q^{\prime}\in[0, 1)\\
        1 & \text{if } q^{\prime}\in\left[1, \qty(\dfrac{A}{A_{\bar{b}}})^{\frac{1}{A_b}}\right]
    \end{cases}
\end{gather*} 
Now call the following condition as condition$(\star)$:
$$\text{$G(q^{\prime}|y)$ is non-decreasing in $y$ for almost all $q^{\prime}\in[0,1]$.}$$
Then 
\begin{gather}\label{eq:condition(*)}
    \text{Condition$(\star)$ holds} \iff \text{$F_{b}^{-1}(\hat{q}|y)$ is non-decreasing in $y$ for almost all $\hat{q}\in[0,1]$}
\end{gather}

is true. In fact,
\begin{itemize}
    \item "$\impliedby$" direction is obvious.
    \item As for "$\implies$" direction, suppose condition$(\star)$ holds. Then since $G(q^{\prime}|y)$ is a constant function when $q^{\prime}\in\left[1, \qty(\dfrac{A}{A_{\bar{b}}})^{\frac{1}{A_b}}\right]$, $G(q^{\prime}|y)$ is non-decreasing in $y$ for almost all $q^{\prime}\in\left[0,\qty(\dfrac{A}{A_{\bar{b}}})^{\frac{1}{A_b}}\right]$. Thus, $F_{b}^{-1}(\hat{q}|y)$ is non-decreasing in $y$ for almost all $\hat{q}\in[0,1]$.
\end{itemize}

But notice that condition$(\star)$ doesn't depend on any of the variables $b, \alpha$ or $n$. Also, condition$(\star)$ is equivalent to that $F$ is quantile-scaled monotone. Thus, if $F$ is quantile-scaled monotone, then for any $\alpha\gg 0$ and any $n\geq 2$, uniform pricing is optimal.

\end{proof}

By above proposition \ref{prop:quantile scaled monotone}, we could show Theorem \ref{thm:when's alpha}. Now we prove the following corollary \ref{cor:two item} appeared in Section \ref{sec:when's alpha}

\begin{cor*}
    Suppose that the profile of relative values is stochastically non-decreasing in positive values of the grand bundle under the setting with $n=2, t_1, t_2\overset{i.i.d.}{\sim}F$.  Then for any $n\ge2$ and $\alpha_1, \ldots, \alpha_n>0$,  uniform pricing remains to be optimal when the seller offers \( n \) items and $t_i$ follows CDF \( F^{\alpha_i} \) for each $i=1, \ldots, n$.
\end{cor*}
\begin{proof}
    This statement is already shown in the proof of Theorem \ref{thm:when's alpha}. In fact, suppose that the profile of relative values is stochastically non-decreasing in positive values of the grand bundle under the setting with $n=2, t_1, t_2\overset{i.i.d.}{\sim}F$. Then, by lemma\ref{lem: strassen quantile}, for any $b\subset B\setminus\{\emptyset, b^*\}$, $F_{b}^{-1}(\hat{q}|y)$ is non-decreasing in $y$ for almost all $\hat{q}\in[0,1]$ \footnote{ The definition of $F_{b}^{-1}(\hat{q}|y)$ is the same as in the proof of Theorem \ref{thm:when's alpha}}. By \eqref{eq:condition(*)}, condition$(\star)$ holds. Therefore, for any $n\ge2$ and $\alpha_1, \ldots, \alpha_n>0$,  uniform pricing remains to be optimal when the seller offers \( n \) items and $t_i$ follows CDF \( F^{\alpha_i} \) for each $i=1, \ldots, n$.
\end{proof}

Finally, we prove Proposition \ref{prop:hashimoto no zyuubunn zyoukenn}. Recall the definition of monotone elasticity.

\begin{dfn*}
    Assume density function $f$ is differentiable. We say $f$ satisfies monotone elasticity condition if the elasticity of $f$, 
    $$\frac{tf'(t)}{f(t)}$$
    is non-decreasing in $t$.
\end{dfn*}

\begin{prop*}
    Assume density function $f$ is differentiable and satisfies monotone elasticity condition. Then its CDF $F$ is quantile-scaled monotone.
\end{prop*}

\begin{proof}
    Assume density function $f$ is differentiable and satisfies monotone elasticity condition. Take any $q\in[0, 1)$. Let $\omega\in [0, 1]$ be the value that satisfies 
    \[
    \frac{F(\omega y)}{F(y)} = q
    \]
    Since $F$ is continuous and strictly increasing on $[0, \bar{v}]$, we will explicitly write $\omega$ as a function of $y$, $\omega:(0, \bar{v}]\to[0, 1]$ given by:
    $$\omega(y)\equiv\frac{F^{-1}(qF(y))}{y}$$
    where $F^{-1}:[0, 1]\to[0, \bar{v}] $ denotes inverse function of $F$. Thus, by inverse function theorem, the derivative of $\omega$ is well-defined on $(0, \bar{v})$ and given by
    $$\frac{d\omega(y)}{dy}=\frac{qF'(y)\cdot y-F^{-1}(qF(y))\cdot F'(F^{-1}(qF(y)))}{y^2\cdot F'(F^{-1}(qF(y)))}$$
    Let $h:(0, \bar{v})\to \mathbb{R}$ be its numerator: 
    $$h(y)\equiv qF'(y)\cdot y-F^{-1}(qF(y))\cdot F'(F^{-1}(qF(y)))$$
    Since $F(0)=0$, $F^{-1}(0)=0$. Thus, $\lim_{y\to +0}h(y)=0$. Moreover
    $$\frac{dh(y)}{dy}=qF'(y)\cdot\left(\frac{y\cdot F''(y)}{F'(y)}-\frac{F^{-1}(qF(y))\cdot F''(F^{-1}(qF(y)))}{F'(F^{-1}(qF(y)))}\right)$$
    Since $q\in[0, 1]$ and $F$ is strictly increasing, $y\ge F^{-1}(qF(y))$ holds. Thus, by monotone elasticity condition, $\frac{dh(y)}{dy}\ge0$ for all $y\in(0, \bar{v})$. This and $\lim_{y\to +0}h(y)=0$ implies $\frac{d\omega(y)}{dy}\ge0$ for all $y\in(0, \bar{v})$. Thus, $F$ is quantile-scaled monotone. 
\end{proof}

\subsection{Omitted proof in Subsection \ref{subsec:analysis of c^star} (Analysis of $c^*$)}\label{proof:analysis of c^star}
\begin{proof}[Proof of Proposition \ref{prop:uniqueness of c_n^*}]
Since $p$ is determined by $c$, let us denote $p$ as $p(c)$. By the derivation of $p(c)$, recall that 
\begin{align}
    p(c)>c \text{ for all } c>0\label{eq:p(c)>c}
\end{align}
holds. 
we want to characterize the value of $c$ that satisfies 
\begin{align}
    (n+1)(c+1-p(c))=c+1\iff p(c)=\frac{n}{n+1}(c+1)\label{eq:p=n(c+1)/(n+1)}
\end{align}
Let
\begin{gather*}
    P(c)\equiv p(c)- \frac{n}{n+1}(c+1)\\
    h(c,p)\equiv 1-(p-c)^n-np(p-c)^{n-1}
\end{gather*}
, then, by the derivation of $p(c)$, $h(c, p(c))=0$ holds for all $c>0 $.  The implicit function theorem tells us that 
\begin{gather}
    p^{\prime}(c)=-\dfrac{\pdv{h(c,p)}{c}}{\pdv{h(c,p)}{p}}=1- \dfrac{p-c}{2(p-c)+(n-1)p} \label{eq:p'(c)}
\end{gather}
Since 
\begin{gather*}
    P'(c)=\qty(1- \dfrac{p-c}{2(p-c)+(n-1)p} )-\dfrac{n}{n+1}=\dfrac{c(n-1)}{(n+1)(2(p-c)+(n-1)p)}>0,
    \end{gather*}
$P(c)$ is strictly increasing in $c\in (0,\infty)$.
Moreover,
\begin{gather*}
    P(n)=p(n)- \frac{n}{n+1}(n+1)=p(n)-n>0 \quad(\because  \eqref{eq:p(c)>c})
\end{gather*}
and 
\begin{gather*}
    P(0)=p(0)-\dfrac{n}{n+1}=\qty(\dfrac{1}{n+1})^\frac{1}{n} - \dfrac{n}{n+1} <0
\end{gather*}
because 
\begin{gather*}
    \qty(\dfrac{1}{n+1})^\frac{1}{n} < \dfrac{n}{n+1} \iff (n+1)^{n-1}<n^n 
\end{gather*}, which always hold when $n\ge2$.\footnote{$(n+1)^{n-1}<n^n$ holds for all $n\geq 2$. Let $f(x)=x\log x-(x-1)\log(x+1)$. Then $f(2)=\log\frac{4}{3}>0$. Further, $x \geq 2$, $f'(x)=\frac{2}{x+1}-\log(1+\frac{1}{x})> \frac{1}{x}-\log(1+\frac{1}{x}) \geq 0$. Therefore, if $x \geq 2$, $f(x)>0$.}
Therefore, $P(c)=0$ has a unique solution $c_n^*$ on $(0, n)$ and $P(c)\le0$ holds on $[0, c_n^*]$. Here, notice that
\begin{align*}
    \{c\in(0, n):P(c)=0\}
    &= \{c\in(0, n):p(c)=\frac{n}{n+1}(c+1)\}\\
    &= \{c\in(0, n):h\left(c, \frac{n}{n+1}(c+1)\right)=0\}\\
    &= \{c\in(0, n):(n+1)^n-(n-c)^n-n^2(c+1)(n-c)^{n-1}=0\}
\end{align*}
holds where the second equality holds since $p$ was uniquely determined as the solution of $h(c, p)=0$ on $(c, c+1)$ and $c<\frac{n}{n+1}(c+1)<c+1\iff c<n$. Thus, $c_n^*$ can be characterized as the unique solution of  
$$(n+1)^n-(n-c)^n-n^2(c+1)(n-c)^{n-1}=0$$ 
on $(0, n)$, too.
\end{proof}

\begin{proof}[Proof of Proposition \ref{prop:c^star strictly increasing}]
    Recall that $c_n^*$ satisfies  $(n+1)(c_n^*+1-p(c_n^*))-(c_n^*+1)=0$. By implicit function theorem,
    $$\frac{dc_n^*}{dn}=-\frac{c+1-p(c)}{(n+1)(1-p'(c))-1}=\frac{c+1-p(c)}{(n+1)p'(c)-n}$$
    By \eqref{eq:p'(c)}, 
    $$(n+1)p'(c)-n=(n+1)\left(1-\frac{p(c)-c}{2(p(c)-c)+(n-1)p(c)}\right)-n=\frac{(n-1)c}{2(p(c)-c)+(n-1)p(c)}$$
    Since we have $c<p(c)<c+1$, the denominator and numerator are both positive. Thus, $\frac{dc_n^*}{dn}>0$.
\end{proof}

\begin{proof}[Proof of Proposition \ref{prop:c_n infty}]
    It suffices to show $c_n^*>\frac{\ln n}{3}$. Since $p(c)$ satisfies $h(c, p(c))=0$, 
    \begin{align*}
        &1=(p(c)-c)^n+np(c)(p(c)-c)^{n-1}\\
        \implies &1>(n+1)(p(c)-c)^n\\
        \iff &p(c)<c+\left(\frac{1}{n+1}\right)^{\frac{1}{n}}
    \end{align*}
    Thus,
    $$P(c)<\frac{1}{n+1}c+\left(\frac{1}{n+1}\right)^{\frac{1}{n}}-\frac{n}{n+1}$$
    Here, let $f(n)=n\ln(n-\frac{\ln n}{3})-(n-1)\ln(n+1)$, then
    \begin{align*}
        f'(n)
        &=\frac{5n-1+(n-1)\ln n}{(3n-\ln n)(n+1)}-\ln(1+\frac{3+\ln n}{3n-\ln n})\\
        &>\frac{2n-4 -2\ln n}{(3n-\ln n)(n+1)}
    \end{align*}
    where the inequality holds since $x>\ln(1+x)$ for all $x>0$. Here, since $2n-4-2\ln n>0$ holds for all $n\ge 4$, $f(n)$ increases for $n\ge 4$. This and $f(4)>0$\footnote{Indeed, since $\frac{3}{4}>\ln2$,  $f(4)=4\ln(4-\frac{2}{3}\ln 2)-3\ln 5>\ln\frac{7^4}{2^4\cdot 5^3}>0$} implies $f(n)>0$ for all $n\ge4$. But notice that
    \begin{align*}
        f(n)>0 \iff \left(n-\frac{\ln n}{3}\right)^n>(n+1)^{n-1}\iff \frac{1}{n+1}c+\left(\frac{1}{n+1}\right)^{\frac{1}{n}}-\frac{n}{n+1} <0
    \end{align*}
    holds. Thus, $P(\frac{\ln n}{3})<0$ for all $n\ge4$. By proof of proposition \ref{prop:uniqueness of c_n^*}, it follows that $c_n^*>\frac{\ln n}{3}$ for all $n\ge4$
\end{proof}
\subsection{Proof of Proposition\ref{prop:uniform pricing is optimal if gamma_1>mu_+}}\label{proof:uniform pricing is optimal if gamma_1>mu_+}
\par

In this subsection, we will prove the stronger statement of Proposition\ref{prop:uniform pricing is optimal if gamma_1>mu_+}. That is, we will prove that we can provide a proper $\gamma$ if a proper measure is defined on a smaller region instead of considering the entire domain of $\gamma$, $[c,c+1]^n$. First, we will introduce an equivalent characterization of the pushed set defined in Definition\ref{dfn:push U} as follows:
\begin{rmk*}[Lemma \ref{lemma:y+t1} restate]
For all upper set $U\subset B$ with respect to $B$ and all $y\in W\cup B_Z$
    \begin{gather*}
        y+t\bm{1} \in U\cap B \text{ for some } t \in \mathbb{R}\iff y\in push(U)
    \end{gather*}
\end{rmk*}

\begin{proof}
($\Rightarrow$)  Suppose that there exists $t\in \mathbb{R}$ such that $y+t\bm{1} \in U\cap B$. By definition of $u^*$, there exists $i\in\{1, ..., n\}$ such that $u^*(y+t\bm{1})-u^*(y)=y_i+t-y_i=t=\|(y+t\bm{1})-y\|_{\infty}$.\\
($\Leftarrow$)
Suppose that there exists $x\in U\cap B$ such that $u^*(x)-u^*(y)=\|x-y\|_{\infty}$.
Take any $i$ such that $x_i=c+1$. Note that 
\[
c+1-\max(y)=u^*(x)-u^*(y)=\|x-y\|_\infty\geq |x_i-y_i|=c+1 -y_i.
\]
This implies that $y_i=\max(y)$.
Let $t=c+1-y_i$. Then $y+t\bm{1}\in B$ because $\max(y+t\bm{1})=y_i+t=c+1$. Here, $y+t\bm{1}\in U$ because $U$ is an upper set and for any $j$,
\begin{eqnarray*}
    y_j+t&\geq &(x_j-x_i+y_i)+t\quad(\because x_i-y_i=c+1-\max(y)=\|x-y\|_\infty\geq x_j-y_j)\\
    &=&(x_j-(c+1)+y_i)+t\\
    &=& x_j
\end{eqnarray*}
\end{proof}

Let
\begin{gather*}
B_i\equiv\qty{x\in[c,c+1]^n: x_i=c+1=\max(x)}
\end{gather*}
Thus $B=\cup_{i=1}^nB_i$. Then, by above lemma\ref{lemma:y+t1},  notice that the restriction of $\gamma_1$ to $B_i$ denoted as $\gamma_1^{B_i}$ can be expressed as follows:
\begin{gather*}
    \gamma_1(E\cap B_i)=\mu_-(\{x \in [c,c+1]^n | x_i=\max(x)\geq p\text{ and } x+t\bm{1} \in E \text{ for some } t \in \mathbb{R}\})
\end{gather*}
for all measurable $E\subset B_i$. Denote $\gamma_1^{B_i}(E)\equiv\gamma_1(E\cap B_i)$. Therefore, equivalent formation of the measure $\gamma_1$ is:
$$\gamma_1(E)=\sum_{i=1}^n\gamma_1^{B_i}(E)$$
for all measurable $E\subset B$. Now we are ready to prove Proposition\ref{prop:uniform pricing is optimal if gamma_1>mu_+}. To do that, it suffices to prove the following Proposition\ref{prop:uniform pricing is optimal if cvx dominance in B_1}:

\begin{prop}\label{prop:uniform pricing is optimal if cvx dominance in B_1}
It is optimal to sell all items at a certain price  if
\begin{gather*}
    \gamma_1^{B_1} \succeq_\text{cvx} \mu|_{B_1},
\end{gather*}
\end{prop}

Notice that if we show Proposition\ref{prop:uniform pricing is optimal if cvx dominance in B_1}, then Proposition\ref{prop:uniform pricing is optimal if gamma_1>mu_+} follows since $\gamma_1$ and $\mu_+$ has measure only on $\{(c, \dots, c)\}\cup B$ and by symmetric structure of $\gamma_1$ and $\mu_+$,
$$\gamma_1^{B_1}\succeq_{cvx}\mu_+|_{B_1}\implies \gamma_1\succeq_{cvx}\mu_+$$
holds.

\begin{proof}[Proof of Proposition \ref{prop:uniform pricing is optimal if cvx dominance in B_1}]
We will prove the existence of $\gamma$ which satisfies the above conditions under $u(x)=\max(0,\max(x)-p)$. \\
The proof is done by dividing $[c,c+1]^n$ into $n+1$ parts as follows:
\begin{gather*}
    A_0\equiv[c,p)^n \\
    A_i= \{x \in [c,c+1]^n | x_i = \max(x) \geq p \}
\end{gather*}
Note that 
\begin{gather*}
    [c,c+1]^n = \qty(\bigcup_{i=0}^n A_i)
\end{gather*}
The proof is done if we provide $\gamma^{A_i}, i=1,\ldots, n$ such that $\gamma^{A_i}$ is defined on $A_i\times A_i$ and 
\begin{align}
  \gamma_1^{A_i}-\gamma_2^{A_i}  &\succeq_\text{cvx} \mu|_{A_i} \tag{2'} \label{eq:cvx-cdn'}\\
  \int_{[c,c+1]^n}ud(\gamma_1^{A_i}-\gamma_2^{A_i})&=\int_{[c,c+1]^n}ud\mu|_{A_i} \tag{3'}\label{eq:mass-cdn'}\\
  u(x)-u(y)&=\|x-y\|_\infty~ \gamma^{A_i}(x,y)\text{-almost surely} \tag{4'}\label{eq:norm-cdn'}
\end{align}
This is because $\gamma$ is given using $\gamma^{A_i},i=1,\ldots, n$ as follows 
\begin{gather*}
\gamma(F)=\left\{
\begin{array}{cc}
    \gamma^{A_i}(F) & \text{ if $F\subset A_i\times A_i$} \\
    0 & \text{ otherwise}
\end{array}
\right.
\end{gather*}
This $\gamma$ satisfies \eqref{eq:cvx-cdn}, \eqref{eq:meas-cdn}, \eqref{eq:norm-cdn}.

\begin{enumerate}[(i)]
    \item Let us define $\gamma^{A_0}$ as
    \begin{gather*}
    \gamma^{A_0} (D)=\left\{
    \begin{array}{ll}
    1 & \text{ if $(\bm{c},\bm{c})\in D$}\\
    0 & \text{ otherwise}
    \end{array}
    \right.
    \end{gather*}
    where $\bm{c}=(c,\ldots,c)\in [c,c+1]^n$.
    Note that $\gamma^{A_0}(A_0^2)=0$.
    \begin{enumerate}
        \item This $\gamma^{A_0}$ satisfies \eqref{eq:cvx-cdn'} because
        \begin{gather*}
            \gamma^{A_0}_1 (F)-\gamma^{A_0}_2(F) = 0 \geq \mu|_{A_0}(F)
        \end{gather*}
        for all measurable $F\subset A_0$.
        \item This $\gamma^{A_0}$ satisfies \eqref{eq:mass-cdn'} because $u(x)=0$ for all $x\in A_0$.
        \item This $\gamma^{A_0}$ satisfies \eqref{eq:norm-cdn'} because $\gamma^{A_0}(F)=0$ if $(\bm{c},\bm{c})\notin F$.
    \end{enumerate}
    \item 
    Let $\gamma^{A_i},i=1,\ldots,n$ be the measures which satisfies
    \begin{gather*}
    \gamma^{A_i}_1(D) = \gamma_1^{B_i}(D\cap B_i)\\
    \gamma^{A_i}_2(D)=\mu|_{B_i}(D)
    \end{gather*}
    for all measurable $D\subset A_i$.
    \begin{enumerate}
        \item This $\gamma^{A_i}$ satisfies \eqref{eq:cvx-cdn'} because
        \begin{gather*}
\gamma_1^{B_1}\succeq_{cvx} \mu|_{B_1}.
        \end{gather*}
        
        \item This $\gamma^{A_0}$ satisfies \eqref{eq:mass-cdn'} because $u(x)$ is constant for all $x\in B_1$.
        \item This $\gamma^{A_0}$ satisfies \eqref{eq:norm-cdn'} because of the definition of $\gamma_1^{B_1}$.
    \end{enumerate}
    
\end{enumerate}

\end{proof}
\subsection{Proof of Sufficiency of Proposition \ref{prop:pushed measure}}\label{proof:sufficiency of thm pushed measure}






Following on from the previous subsection, we have to seek the condition under which
$\gamma_1^{B_1} \succeq_\text{cvx} \mu_+|_{B_1}$ holds. Before we do that, we will identify density functions of $\gamma_1^{B_1}$ and $\mu_{+}|_{B_1}$. It is straightforward that $\mu_{+}|_{B_1}$'s density is $f_{\mu_+}|_{B_1}(x)\equiv c+1$. On the other hand, the following lemma$\ref{lemma:density}$ shows the density of $\gamma_1^{B_1}$ is  written as 
\begin{eqnarray}
f^{B_1}(x)\equiv
  \begin{cases}
    (n+1)(c+1-p) & (x \in [2c+1-p,c+1]^{n-1} ) \\
    (n+1)(\min(x)-c)+c & (\text{otherwise})
  \end{cases}
\end{eqnarray}
\begin{rmk}\label{lemma:density}
    \begin{gather*}\gamma_1^{B_1}(E)=\int_{[c,c+1]^{n-1}} \mathbb{I}_E(c+1,x_{-1})f^{B_1}(x_{-1})dx_{-1}
\end{gather*} for all measurable set $E\subset B_1$.
\end{rmk}
\begin{proof}
    Let $\mathbb{L}$ be the set of all rectangles on $B_1$, i.e. $\mathbb{L}\equiv\{\{c+1\}\times[a_2, b_2]
    \times\cdots\times[a_n, b_n]\subset B_1:a_2\le b_2\dots a_n\le b_n\}$
    Since product $\sigma$-algebra on $B_1$ is generated by $\mathbb{L}$, it suffices to show that for each $L\in\mathbb{L}$, 
    \begin{align}
        \gamma_1^{B_1}(L)=\int_{[c,c+1]^{n-1}}\mathbb{I}_L(c+1,x_{-1})f^{B_1}(x_{-1})dx_{-1}\label{eq:rectangle}
    \end{align}
    holds. Here define $\mathbb{L}^*$ as the set of all rectangles on $B_1$ which has one corner at $x=(c+1, c, \dots, c)$, i.e. $\mathbb{L}^*\equiv\{\{c+1\}\times [c, c+l_2]\times\cdots\times[c, c+l_n]:l_2, \dots l_n\in[0, 1]\}$. Then, notice that each $L=\{c+1\}\times[a_2, b_2]
    \times\cdots\times[a_n, b_n]\in\mathbb{L}$ can be expressed as the finite union and intersection of elements of $\mathbb{L}^*$ as follows:\; if you let $L_0=\{c+1\}\times [c, b_2]\times\cdots\times[c, b_n]$ and $ L_i=\{x\in B_1: x_i\in[c, a_i]\}$ for each $i=2, \dots, n$, then 
    $$L=L_0\bigcap (\cap_{i=2}^n(B_1\setminus L_i))$$
    Thus, instead of showing \eqref{eq:rectangle} for each $L\in\mathbb{L}$, it suffices to show that $\eqref{eq:rectangle}$ for each $L=\{c+1\}\times [c, c+l_2]\times\cdots\times[c, c+l_n]\in\mathbb{L}^*$.
    By symmetry, we can assume $l_2\ge l_3\cdots\ge l_n$ without loss of generality.
    We will divide it into two cases.
    \begin{enumerate}
        \item When $c+1-l_n\ge p$\par
        Let
        $$L_i\equiv L\cap\{x:x_i=\min(x)\},$$
        then, $(L_2, \dots, L_n)$ is a partition of $L$, i.e. these are mutually disjoint and $L=\cup_{i=2}^nL_i$. Here,
        \begin{align*}
            \gamma_1^{B_1}(L)
            &=\mu_-(\{x \in [c,c+1]^n | x_1=\max(x)\geq p, \text{ and } x+t\bm{1} \in L\text{ for some }t \in \mathbb{R} \})\\
            &=\int_{0}^{l_n}(l_2-t)\cdots(l_n-t)(n+1)dt\\
            &\quad+\sum_{i=2}^n\mu_-(\{x \in [c,c+1]^n | x_1=\max(x)\geq p, x_i=c, \text{ and } x+t\bm{1} \in L_i\text{ for some } t \in \mathbb{R} \})
        \end{align*}
        As for the first term, we can calculate it as
        \begin{align*}
            &\int_{0}^{l_n}(l_2-t)\cdots(l_n-t)(n+1)dt\\
            &=[t(l_2-t)\cdots(l_n-t)]_{0}^{l_n}+\sum_{i=2}^n
            (n+1)\int_0^{l_n}(l_2-t)\cdots(l_{i-1}-t)(l_{i+1}-t)\cdots(l_n-t)t\;dt\\
            &=\sum_{i=2}^n
            (n+1)\int_0^{l_n}(l_2-t)\cdots(l_{i-1}-t)(l_{i+1}-t)\cdots(l_n-t)t\;dt
        \end{align*}
        As for the second term, since $L_i=L\cap \{x:x_i=min(x)\}$, notice that
        \begin{align*}
            &\{x \in [c,c+1]^n | x_1=\max(x)\geq p, x_i=c, \text{ and } x+t\bm{1} \in L_i\text{ for some }t \in \mathbb{R}\}\\
            &=\{x \in [c,c+1]^n | x_1=\max(x)\geq p, x_i=c, \text{ and } x+(c+1-x_1)\bm{1} \in L_i\}\\
            &=\bigcup_{u=c+1-l_n}^{c+1}\{x \in [c,c+1]^n |x_1=u,  x_i=c, \text{ and } x+(c+1-u)\bm{1} \in L_i\}
        \end{align*}
        Notice that for each $u\in[c+1-l_n, c+1]$, 
        \begin{align*}
            &\{x+(c+1-u)\bm{1}\in L_i:x\in[c, c+1]^n, x_1=u, x_i=c\}\\
            &=\{x\in L_i:x_1=c+1, x_i=2c+1-u, x_j\in[2c+1-u, c+l_j] \text{ for all }j\not=i\}
        \end{align*}  
        holds. Thus, each component of the second term can be calculated as
        \begin{align*}
            &\mu_-(\{x \in [c,c+1]^n | x_1=\max(x)\geq p, \text{ and } x+t\bm{1} \in L\text{ for some }t \in \mathbb{R}\})\\
            &=\int_{c+1-l_n}^{c+1}\left(\prod_{j\not=1, i}\left(c+l_j-(c+1-u)-((2c+1-u)-(c+1-u))\right)\right)  du\\
            &=\int_{c+1-l_n}^{c+1}\left(\prod_{j\not=1, i}\left(l_j-(c+1-u)\right)\right)  du\\
            &=\int_c^{c+l_n}\prod_{j\not=1, i}(c+l_j-v)dv
        \end{align*}
        where I changed the variable as $u=-v+2c+1$ in the final equation.  On the other hand,
        \begin{align*}
            \int_{[c,c+1]^{n-1}}\mathbb{I}_L(c+1,x_{-1})f^{B_1}(x_{-1})dx_{-1}
            &=\int_{L}(n+1)(\min(x)-c)+c\;dx_{-1}\\
            &=\sum_{i=2}^n \int_{L_i}(n+1)(\min(x)-c)+c\;dx_{-1}\\
            &=\sum_{i=2}^n \int_{L_i}(n+1)(x_i-c)\;dx_{-1}+c\int_{L_i}dx_{-1}
        \end{align*}
        As for the first term, for each $i\ge2$
        \begin{align*}
            \int_{L_i}(n+1)(x_i-c)\;dx_{-1}
            &=(n+1)\int_{x_i}^{c+l_2}\cdots\int_{x_i}^{c+l_{i-1}}\int_{c}^{c+l_n}\int_{x_i}^{c+l_{i+1}}\cdots\int_{x_i}^{c+l_n}(x_i-c)\;dx_{-1}\\
            &=(n+1)\int_0^{l_n}(l_2-t)\cdots(l_{i-1}-t)(l_{i+1}-t)\cdots(l_n-t)t\;dt
        \end{align*}
        where I transformed the integrating variable as $x_i=t+c$ in the final equation. 
        As for the second term, since $L_i= L\cap\{x:x_i=\min(x)\},$ for each $i\ge2$
        $$\int_{L_i}dx_{-1}=\int_c^{c+l_n}\prod_{j\not=1, i}(c+l_j-x_i)dx_i$$
        Thus, in this case, we could show $\eqref{eq:rectangle}$.
    
        
        \item When $c+1-l_n< p$\\
        For each $S\subset N_{-1}\equiv\{2, 3, \dots, n\}$, let 
        $$L_S\equiv\{x\in\mathbb{R}^{n+1}:x_1=c+1, x_i\in[c, 2c+1-p] \text{ for all $i\in S$ and } x_j\in(2c+1-p, c+l_j]\text{ for all $j\in N_{-1}\backslash S$}\}$$
        then, $(L_S)_{S\subset N_{-1}}$ is a partition of $L$. 
        To prove $\eqref{eq:rectangle}$, it suffices to show
        $$\gamma_1^{B_1}(L_{\emptyset})=\int_{[c,c+1]^{n-1}}\mathbb{I}_{L_{\emptyset}}(c+1,x_{-1})f^{B_1}(x_{-1})dx_{-1}$$
        The reason is as follows: The left-hand side of $\eqref{eq:rectangle}$ is
        $$\gamma_1^{B_1}(L)            =\sum_{S\subset N_{-1}}\gamma_1^{B_1}(L_S)$$
        and the right hand side of $\eqref{eq:rectangle}$ is    
        $$\int_{[c,c+1]^{n-1}} \mathbb{I}_L(c+1,x_{-1})f^{B_1}(x_{-1})dx_{-1}\\
            =\sum_{S\subset N_{-1}}\int_{[c,c+1]^{n-1}} \mathbb{I}_{L_S}(c+1,x_{-1})f^{B_1}(x_{-1})dx_{-1}$$
        Notice that by the result from case $(1)$, 
        \begin{align*}                          
            \gamma_1^{B_1}(L_{N_{-1}})
            &=\int_{[c,c+1]^{n-1}} \mathbb{I}_{L_{N_{-1}}}(c+1,x_{-1})f^{B_1}(x_{-1})dx_{-1}\\      \gamma_1^{B_1}(L_{N_{-1}})+\gamma_1^{B_1}(L_S)
            &=\int_{[c,c+1]^{n-1}} \left(\mathbb{I}_{L_{N_{-1}}}(c+1,x_{-1})+\mathbb{I}_{L_S}(c+1,x_{-1})\right)f^{B_1}(x_{-1})dx_{-1} \\  &\quad\text{ for all $S\in2^{N_{-1}}\backslash \{N_{-1}, \emptyset\}$}
        \end{align*}
        holds. Thus, to prove $\eqref{eq:rectangle}$ it remains to show that 
        $$\gamma_1^{B_1}(L_{\emptyset})            
        =\int_{[c,c+1]^{n-1}}\mathbb{I}_{L_{\emptyset}}(c+1,x_{-1})f^{B_1}(x_{-1})dx_{-1}$$
        But each of these are:
    \begin{eqnarray*}
        &\quad&\gamma_1^{B_1}(L_{\emptyset})   \\      
        &\quad&\quad=\mu_-(\{x \in [c,c+1]^n | x_1=\max(x)\geq p, \text{ and } x+t\bm{1} \in L_{\emptyset} \text{ for some }t \in \mathbb{R}~\})\\
        &\quad&\quad=(c+l_2-(2c+1-p)) \cdots (c+l_n-(2c+1-p))(c+1-p)(n+1)\\
        &\quad&\int_{[c,c+1]^{n-1}}\mathbb{I}_{L_{\emptyset}}(c+1,x_{-1})f^{B_1}(x_{-1})dx_{-1}\\
        &\quad&\quad=(c+l_2-(2c+1-p)) \cdots (c+l_n-(2c+1-p))(c+1-p)(n+1)
    \end{eqnarray*}
        So, we could also show $\eqref{eq:rectangle}$ in this case.
    \end{enumerate}   
\end{proof}
For each $k\in \mathbb{N}$, define $D_k\subset\mathbb{R}^k$ as follows:
\begin{gather*}
   D_k\equiv\{x\in[c, c+1]^k: x_1=c+1\ge x_2\ge\cdots\ge x_{k}\}
\end{gather*}
Using above $D_k$, consider the restriction of $\gamma_1^{B_1}$ to $D_n$ (denoted as $\gamma_1^{D_n}$) and its corresponding density function (denoted as $f^{D_n}$) as follows:
\begin{gather*}
   \gamma_1^{D_n}\equiv\gamma_1^{B_1}|_{D_n}\\
   f^{D_n}\equiv f^{B_1}|_{D_n}
\end{gather*}
Notice that $\gamma_1^{D_n}$ satisfies the condition in Lemma \ref{lemma:Kash and Frongillo} in $D_n$. Since density functions $f_{\mu_+}|_{B_1}$ and $f^{B_1}$ have symmetric structure, to show $\gamma_1^{B_1} \succeq_\text{cvx} \mu_+|_{B_1}$, it suffices to show $\gamma_1^{D_n} \succeq_\text{cvx} \mu_+|_{D_n}$.
Now we are ready to show the simple condition under which
$\gamma_1^{D_n} \succeq_\text{cvx} \mu_+|_{D_n}$ holds:
\begin{prop}\label{prop:iff cdn for cvx dominance}For all $n\geq 1$,
$$\gamma_1^{D_n} \succeq_\text{cvx} \mu_+|_{D_n} \iff (n+1)(c+1-p) \geq c+1 $$
\end{prop}
\begin{proof}
    $\implies$ part is obvious by considering contraposition. Indeed, suppose $(n+1)(c+1-p) < c+1 $.Recall we defined
    \[
    v(x) \equiv \max\qty(\sum_{i=1}^n x_i - (nc+n+c-p), 0).
    \]  
    Since  
    \[
    v(c+1, \dots, c+1, 2c+1-p) = 0,
    \]  
    it follows that \( v(x) < 0 \) whenever \( \max(x) < 2c + 1 - p \), leading to  
    \begin{gather*}
        \int_X v d\gamma_1 < \int_X v d\mu_{+} 
    \end{gather*}
    In particular,  we have
    \[
    \int_X v d\mu_+|_{D_n} >\int_X vd\gamma_1^{D_n}.
    \]
    Since $v\in \mathcal{U}(X)$, $\gamma_1^{D_n} \succeq_\text{cvx} \mu_+|_{D_n}$ doesn't hold.
    \par
    $\impliedby$ part will be shown by the following Proposition \ref{prop:D-ordered->cvx dominance}  
\end{proof}
\begin{dfn}\label{dfn:D-ordered}
    For $k\ge2$ and $X\subset\mathbb{R}^{k-1}$, we say the function $f:X\to\mathbb{R}$ is \textbf{X-ordered} if it satisfies the following two conditions:
    \begin{itemize}
        \item $f(x)=f(x')\;$ if $\min(x)=\min(x')$
        \item $\text{ there exists } r\in\mathbb{R} \text{ such that }  f(x)(\min(x)-r)\ge0$ for all $x\in X$
    \end{itemize}

\end{dfn}

\begin{prop}\label{prop:D-ordered->cvx dominance}
    Let $\alpha$ be a signed measure on $X$ with a density function $f$. If $f$ is $X$-ordered and $\alpha(X)\geq 0$, then $\alpha\succeq_{\text{cvx}} 0$.
\end{prop}

Proposition \ref{prop:D-ordered->cvx dominance} is stronger than the $\impliedby$ part of Proposition \ref{prop:iff cdn for cvx dominance} because $\gamma_1^{D_n} -\mu_+|_{D_n}$ has a density function, and its density function is $D_n$-ordered if $(n+1)(c+1-p)\ge(c+1)$ by following lemma:
\begin{rmk}
    If $(n+1)(c+1-p)\ge(c+1)$, the density function of $\gamma_1^{D_n} -\mu_+|_{D_n}$ is $D_n$-ordered 
\end{rmk}
\begin{proof}
    The first condition is straightforward because $f^{D_n}$ depends only on $\min(x)$ and $f_{\mu_+}|_{D_n}$ is constant. The second condition is also satisfied. Let $r\in \mathbb{R}$ satisfy
\begin{gather*}
    (n+1)(r-c)+c=c+1,
\end{gather*} which means $r=\dfrac{1}{n+1}+c$.
Then $(f^{D_n}-f_{\mu_+}|_{D_n})(x)(\min(x)-r)\geq0$ for all $x\in D_n$ as follows:
\begin{enumerate}
    \item If $x\in \{c+1\}\times[2c+1-p,c+1]^{n-1}$, $(f^{D_n}-f_{\mu_+}|A)(x)(\min(x)-r)=((n+1)(c+1-p)-(c+1))(\min(x)-r)$. 
    Since we assume $(n+1)(c+1-p)\geq (c+1)$,  \begin{gather*}
        (2c+1-p)-r=(2c+1-p)-\qty(\dfrac{1}{n+1}+c)=\dfrac{(n+1)(c+1-p)-(c+1)}{n+1}\geq 0.
    \end{gather*} Then$(2c+1-p)\geq r$. 
    Therefore, $((n+1)(c+1-p)-(c+1))(\min(x)-r)\geq0$ for all $x\in [2c+1-p,c+1]^{n-1}$.
    \item If $x\in  D_n\backslash\{c+1\}\times[2c+1-p,c+1]^{n-1}$, 
    \begin{eqnarray*}
        (f^{D_n}-f_{\mu_+}|_{D_n})(x)(\min(x)-r)&=&((n+1)(\min(x)-c)-(c+1))(\min(x)-r)\\&=&(n+1)(\min(x)-r)^2\geq 0
    \end{eqnarray*}
    by definition of $r$.  
\end{enumerate}
Therefore, $(f^{D_n}-f_{\mu_+}|_{D_n})(x)(\min(x)-r)\geq0$ for all $x\in D_n$. 
Thus, $f^{D_n}-f_{\mu_+}|_{D_n}$ is $D_n$-ordered. 
\end{proof}
To prove Proposition \ref{prop:D-ordered->cvx dominance},  we will first introduce the sufficient condition of convex dominance.

\begin{rmk}[\textcite{shaked2007stochastic}, section 6.B.1]\label{lemma:Shaked and Shanthikumar}
    For $k\ge2$ and $D\subset \mathbb{R}^k$, let $\nu$ and $\eta$ be the signed measure supported within $D$. Then $\nu\succeq_{cvx}\eta$ if $\int \mathbb{I}_U(x)d\nu\ge\int \mathbb{I}_U(x)d\eta$ for all upper set $U\subset \mathbb{R}^k$.
\end{rmk}
Thus, to show $\alpha\succeq_\text{cvx} 0$, we will show $\int \mathbb{I}_U(x)d\alpha\ge 0$ for all upper set $U\subset D_n$. 
we will use mathematical induction to prove Proposition \ref{prop:D-ordered->cvx dominance}. Before proceeding to the proof of Proposition \ref{prop:D-ordered->cvx dominance}, we introduce two lemmas useful for conducting induction.

\begin{rmk}\label{lemma:D_k-ordered->D_k-1-ordered}
    For $k\ge3$, $c\ge0$, and $k$-dimensional euclidean space, consider integrable function $f:D_k\to\mathbb{R}$. Here, define the function on $(k-1)$-dimensional space; $g:D_{k-1}\to\mathbb{R}$ given as $g(x)\equiv\int_c^{x_{k-1}}fdx_k$. If $f$ is $D_k$-ordered, then $g$ is $D_{k-1}$-ordered.
\end{rmk}
\begin{proof}
    Take any $x, x'\in D_{k-1}$ with $x_{k-1}=x_{k-1}'$, then since $f$ is $D_k$-ordered,  
    \begin{gather*}
        g(x)=\int_c^{x_{k-1}}f(x, x_k)dx_k=\int_c^{x_{k-1}'}f(x', x_k)dx_k=g(x').
    \end{gather*} The first condition is proved.\\
    Let $r\in\mathbb{R}$ be a number such that $f(x)(x_k-r)\ge0$ for all $x\in D_k$. Note that we can further assume such $r$ to be in $[c, c+1]$. Take any $x\in D_{k}$ with $x_{k-1}\in[c, r)$; then since $x_k<r$, $g(x_1, \dots, x_{k-1})=\int_c^{x_{k-1}}fdx_k\le0$. Take any $x', x''\in D_k$ with $x_{k-1}''> x_{k-1}'\ge r$, then
    \begin{eqnarray*}
        &\quad&g(x_1'', \dots, x_{k-1}'')-g(x_1', \dots, x_{k-1}')\\
        &\quad&=\int_c^{x_{k-1}''}f(x_1'', \dots, x_{k-1}'', x_k)dx_k-\int_c^{x_{k-1}'}f(x_1', \dots, x_{k-1}', x_k)dx_k\\
        &\quad&=\int_{x_{k-1}'}^{x_{k-1}''}f(x_1'', \dots, x_{k-1}'', x_k)dx_k\ge0.
    \end{eqnarray*} Thus, $g(x)$ is increasing on $x_{k-1}\in[r, c+1]$. From the above two cases, it follows that there exists $ r'\in\mathbb{R}, \; g(x)(x_{k-1}-r')\ge0$ for all $x\in D_{k-1}$.
\end{proof}

\begin{dfn}
\begin{enumerate}
    \item For any $V\subset \mathbb{R}^{n-1}$, the cylinder of $V$, $cy(V)$, is defined as
     \begin{gather*}
         cy(V)\equiv\qty{(x_{-n},x_{n})\in \mathbb{R}^{n}: x_{-n}\in V}
     \end{gather*}
     \item Let $f$ be $X$-ordered. For any $V\subset X$, the cylinder of $V$ by $f$, $cy(V,f)$, is defined as
     \begin{gather*}
         cy(V,f)\equiv\qty{(x_{-n},x_n)\in D: (x_{-n},r)\in V}, 
     \end{gather*}
     where $r$ is determined by the second condition of Definition \ref{dfn:D-ordered}.
\end{enumerate}
     
\end{dfn}

\begin{rmk}\label{lemma:cylinder}
    A signed measure $\alpha$ on $X$ is X-ordered if $\alpha$ has a density function $f:D_n \rightarrow \mathbb{R}$ and $f$ is $D_n$-ordered. For any upper set $U\subset D_n$, $\alpha(cy(U,f))\leq \alpha(U)$.
\end{rmk}
\begin{proof}
Since $f$ is $D_n$-ordered, there exists $r\in\mathbb{R}$ such that $f(x)(\min(x)-r)\geq 0$ for all $x\in D_n$.
Let 
\begin{gather*}
    U^{+}\equiv \qty{x\in U_n: \min(x)\geq r}\\
    U^{-}\equiv \qty{x\in U_n: \min(x)\leq r}\\
    cy(U,f)^{+}\equiv \qty{x\in cy(U,f): \min(x)\geq r}\\
    cy(U,f)^{-}\equiv \qty{x\in cy(U,f): \min(x)\leq r}
\end{gather*}
It suffices to prove that
\begin{gather*}
    cy(U,f)^{+}\subset U^{+}\text{ and }
      U^{-}\subset cy(U,f)^{-}.
\end{gather*}
First, we will prove that $cy(U,f)^{+}\subset U^{+}$. Take any $x\in cy(U,f)^{+
}$. By definition of the cylinder, $(x_{-n},r)\in U$. The 
 definition of $ cy(U,f)^{+}$ implies that $\min(x)\geq r$. Since $x\in cy(U)\subset D_n$, $\min(x)=x_n$. Therefore, 
 $(x_{-n},r)\leq x$ coordinate-wise. By definition of the upper set, $x\in U$.\\
 Second, we will prove $U^{-}\subset cy(U,f)^{-}$. Take any $x\in U^{-}$. Since $x_n=\min(x)\leq r$, $(x_{-n},r)\in U$. By definition of $cy(U)$, $x\in cy(U,f)$. Since $\min(x)\leq r$, $x\in cy(U,f)^{-}$.Therefore, the above statement is true.\\
 Since $f(x)\geq 
 0 $ if $\min(x)\geq r$ and $f(x)\leq 
 0 $ if $\min(x)< r$,
 $\alpha(cy(U,f)^{+})\geq  \alpha(U^{+})\text{ and }
      \alpha(cy(U,f)^{-})\geq \alpha(U)$.
Therefore, 
\begin{gather*}
\alpha(cy(U,f))=\alpha(cy(U,f)^{+})+\alpha(cy(U,f)^{-})\leq \alpha(U^{+})+\alpha(U^{-}) =\alpha(U)
\end{gather*}

\end{proof}

Using Lemma \ref{lemma:D_k-ordered->D_k-1-ordered} and \ref{lemma:cylinder}, we will complete the proof of Proposition \ref{prop:D-ordered->cvx dominance}. We will prove Proposition \ref{prop:D-ordered->cvx dominance} by induction.
\begin{proof}
\begin{enumerate}
    \item When $n=0$, the domain of $\alpha$, $X$, is 
 a singleton. Since $\alpha(X)\geq0$,  $\alpha\succeq_{\text{cvx}}0$.
    \item Suppose that $\alpha\succeq_\text{cvx} 0$ holds for $n=k$. We will prove that this is also true when $n=k+1$.
    By Lemma \ref{lemma:Shaked and Shanthikumar}, it suffices to prove $\alpha(U)\geq 0$ for all upper set $U\subset \mathbb{R}^{k+1}$.
    By Lemma \ref{lemma:D_k-ordered->D_k-1-ordered}, $g:D_{k}\rightarrow \mathbb{R}$ defined as
\begin{gather*}
    g(x)\equiv  \int^{x_{k}}_c f(x)dx_{k+1}
\end{gather*}
is $D_{k}$-ordered. 
Based on the induction hypothesis, $g\succeq_{cvx}0$. Then, 
\begin{gather*}
    \int_{U} g(x)dx \geq 0
\end{gather*}
 for all upper set $U\subset \mathbb{R}^{k}$.
 Note that 
 \begin{gather*}
     \int_{U} g(x)dx = \int_{U} \int^{x_{k}}_c f(x)dx_{k+1} dx_{-(k+1)} = \int_{cy(U)} f(x)dx=\alpha(cy(U))\geq 0. 
 \end{gather*}
 By Lemma \ref{lemma:cylinder}, $\alpha(U)\geq\alpha(cy(U))\geq 0$.
Therefore,  $\alpha\succeq_{\text{cvx}}0$.
    \end{enumerate}
    \end{proof}

Then, $\impliedby$ part of Proposition\ref{prop:iff cdn for cvx dominance} is proven as a corollary of Proposition \ref{prop:D-ordered->cvx dominance}. 

\subsection{Proof of Necessity of Proposition \ref{prop:pushed measure}}\label{proof:necessity of thm pushed measure}


\begin{thm*}[Proposition \ref{prop: nescessity} restate]
    Let $c>c^* $. There doesn't exist $\delta\in \Gamma_{+}(X\times X)$ such that
\begin{gather}
    \mu_+(X)=\delta_1(X)=\delta_2(X),\tag{\ref{eq:delta moreover cdn}}\\
    \delta_1\succeq_{cvx} \mu_{+},\tag{\ref{eq:delta_1 cvx cdn}}\\
    \mu_{-}\succeq_{cvx} \delta_2,\tag{\ref{eq:delta_2 cvx cdn}}\\ 
    u^*(x)-u^*(y)=\|x-y\|_{\infty}\quad\delta(x,y)\text{-almost surely},\tag{\ref{eq:delta a.s. cdn}}\\
    \int_{X}u^*d(\delta_1-\delta_2)=\int_{X}u^*d\mu.\tag{\ref{eq:delta integral cdn}}
    \end{gather}
\end{thm*}

To prove this theorem, we establish the following lemmas. Recall that
\begin{gather*}
    W \equiv \qty{x\in X: \max(x)\in (p,c+1]}\\
    Z \equiv \qty{x\in X: \max(x)\in [c,p]}\\
    B\equiv \qty{x\in X: \max(x)=c+1}\\
    B_Z \equiv \qty{x\in X: \max(x)=p}
\end{gather*}

\begin{rmk*}[Lemma \ref{lemma:both delta integral equal} restate]
Suppose that there exists $\delta\in \Gamma_{+}(X\times X)$ that satisfies  \eqref{eq:delta_1 cvx cdn}, \eqref{eq:delta_2 cvx cdn} and \eqref{eq:delta integral cdn}.
Then 
\begin{align*}
    \int_{X}u^*d\delta_1=\int_{X}u^*d\mu_+,\\
    \int_{X}u^*d\delta_2=\int_{X}u^*d\mu_-.
\end{align*}
\end{rmk*}
\begin{proof}
Since $u^*$ is continuous, non-decreasing, and convex, \eqref{eq:delta_1 cvx cdn} and \eqref{eq:delta_2 cvx cdn} mean
\begin{gather*}
    \int_{X}u^*d\delta_1\geq \int_{X}u^*d\mu_+,\\
    \int_{X}u^*d\delta_2\leq \int_{X}u^*d\mu_-.
\end{gather*}
Here, \eqref{eq:delta integral cdn} implies that both inequalities hold with equality.

\end{proof}

\begin{rmk*}[Lemma\ref{lemma:divide Z and W} restate]
Suppose that there exists $\delta\in \Gamma_{+}(X\times X)$ that satisfies \eqref{eq:delta moreover cdn}, \eqref{eq:delta_1 cvx cdn}, \eqref{eq:delta_2 cvx cdn} and, \eqref{eq:delta integral cdn}.
\begin{gather*}
    \delta_1(Z)=\mu_+(Z)\\
    \delta_1\left(B\right)=\mu_+\left(B\right)
\end{gather*}
In particular, the second equation and \eqref{eq:delta moreover cdn} implies that $\delta_1$ that satisfies \eqref{eq:delta moreover cdn}, \eqref{eq:delta_1 cvx cdn}, \eqref{eq:delta_2 cvx cdn} and, \eqref{eq:delta integral cdn} has measure only on $Z\cup B$.
\end{rmk*}

\begin{proof}
$\mu_{+}(B)\geq \delta_1(B)$ follows from 
\begin{eqnarray*}  
(c+1-p)\mu_{+}(B)
&=&\int_{B}u^*d\mu_{+}\\
&=&\qty(\int_{X}u^*d\delta_1-\int_{X}u^*d\mu_{+})+\int_{B}u^*d\mu_{+}\quad (\because \text{Lemma} \ref{lemma:both delta integral equal})\\
&=& \int_{X}u^*d\delta_1-\int_{X\backslash B }u^*d\mu_{+}\\
&\geq& \int_{B}ud\delta_1\quad \qty(\because \int_{X\setminus B}u^*d\mu_+ = 0)\\
&=&(c+1-p)\delta_{1}(B).
\end{eqnarray*}

Next, we will prove $\mu_{+}(B)\leq \delta_1(B)$.
Let $h_t(x)\equiv\max(t\max(x)-t(c+1)+1, 0)$, where $t>0$. Then, $h_t\leq 1$ for all $t>0$,  and $h_t$ pointwise converges to 
\begin{gather*}
    h(x)\equiv
    \left\{\begin{array}{ll}
    1 & (x\in B) \\
    0  & (x\in X\backslash B).
\end{array}\right.
\end{gather*}

Notice that $h_t$ is a maximum of finite affine functions for each $t>0$. Thus,  $h_t\in\mathcal{U}(X)$ and by \eqref{eq:delta_1 cvx cdn},
\begin{equation*}
    \int_{X} h_t \, d\delta_1 \geq  \int_{X} h_t \, d\mu_+.
\end{equation*}

By the bounded convergence theorem,
\begin{gather*}
    \lim_{t\rightarrow \infty} \int_{X} h_t d\delta_1 = \int_{X}  h d\delta_1 =\delta_1(B)\\
    \lim_{t\rightarrow \infty} \int_{X} h_t d\mu_{+} = \int_{X}  h d\mu_{+}=\mu_{+}(B)
\end{gather*}
Hence, $\mu_{+}(B)\leq \delta_1(B)$.
Thus, we obtain $\delta_1\left(B\right)=\mu_+\left(B\right)$.\par

Next, we will prove that $\delta_1(Z)= \mu_+(Z)$. By $\delta_1\left(B\right)=\mu_+\left(B\right)$ and Lemma \ref{lemma:both delta integral equal},  
$$
\int_{W\setminus B}u^*d\delta_1
=\int_{X}u^*d\delta_1-\int_{B}u^*d\delta_1
=\int_{X}u^*d\mu_{+}-\int_{B}u^*d\mu_{+}
=\int_{W\setminus B}u^*d\mu_{+}=0
$$
This implies that $u^*=0$ $\delta_1$-almost surely in $W\backslash B$, which means that $\delta_1(W\backslash B)=0$.
Therefore,
$$
\delta_1(Z)=\delta_1(X)-\delta_1(W\backslash B)-\delta_1(B)=\mu_{+}(X)-\mu_{+}(W\backslash B)-\mu_{+}(B)=\mu_{+}(Z)
$$
\end{proof}


\begin{rmk*}[Lemma \ref{lemma:pushu} restate]
Suppose that there exists $\delta\in \Gamma_{+}(X\times X)$ that satisfies \eqref{eq:delta a.s. cdn}.
Then the following two equations hold:
\begin{gather*}
    \delta_1(Z\setminus B_Z)=\delta_2(Z\setminus B_Z)\\
    \delta_1(U\cap W)\leq \delta_2(push(U)\cap W)+(\delta_2-\delta_1)(push(U)\cap B_Z) \text{ for  all upper set $U\subset B$ with respect to $B$}.
\end{gather*}
\end{rmk*}

\begin{proof}
Let $S_{u^*}\equiv \qty{(x,y)\in X\times X |u(x)-u(y)=\|x-y\|_{\infty}}$ and $T_{u^*}=(X\times X)\setminus S_{u^*}$. According to \eqref{eq:delta a.s. cdn}, $\delta(T_{u^*})=0$. Note that for all $P,Q\subset X$, 
\begin{gather*}
     [x\in P  \Rightarrow y\in Q \text{ for all $(x,y)\in S_{u^*}$ }] \Rightarrow \delta_1(P)\leq \delta_2(Q)\\
     [x\in P  \iff y\in Q \text{ for all $(x,y)\in S_{u^*}$ }] \Rightarrow \delta_1(P)= \delta_2(Q)
\end{gather*}  This is because, if $x\in P  \Rightarrow y\in Q \text{ for all $(x,y)\in S_{u^*}$ }$,  
 \begin{eqnarray*}
    \delta_1(P)
    &=&\delta(P\times X)\\
    &=& \delta(P\times Q)+\delta(P\times(X\setminus Q))\\
    &=& \delta(P\times Q) (\because P\times(X\setminus Q)\subset T_{u^*} \text{ and \eqref{eq:delta a.s. cdn}})\\
    &\leq& \delta(X\times Q)\\
    &=&\delta_2(Q).
 \end{eqnarray*}
 We obtain the second statement by repeating the above discussion twice. 
 
First, we will prove that for all $(x,y)\in S_{u^*}$,
\begin{gather}\label{eq:eqivalentZBZ}
    x\in Z\setminus B_Z  \iff y\in  Z\setminus B_Z  
\end{gather}
\begin{itemize}
\item If $x\in Z\setminus B_Z$, $u^*(x)=0$. Then, $\|x-y\|_{\infty}=0$ because 
\begin{gather*}
    0\geq 0-u^*(y)=u^*(x)-u^*(y)=\|x-y\|_{\infty}\geq 0.
\end{gather*} Thus, $y=x\in Z\setminus B_Z$.
    \item Suppose that $x\notin Z\setminus B_Z$ and $y\in Z\setminus B_Z$. We will prove that $(x,y)\notin S_{u^*}.$ Since $x\in  W\cup B_Z$, by  definition of $u^*$, $u^*(x)-u^*(y)=(\max(x)-p)-0=\max(x)-p$.
Take any $i^*\in\qty{1,\ldots,n}$ such that $x_{i^*}=\max(x)$. Then 
\begin{eqnarray*}
    \|x-y\|_{\infty}
    &\geq& |x_{i^*}-y_{i^*}|\\
    &=&x_{i^*}-y_{i^*}\quad (\because x_{i^*}\geq p > y_{i^*})\\
    &>&\max(x)-p\quad(\because y_{i^*}<p)\\
    &=& u^*(x)-u^*(y).
\end{eqnarray*} 
Therefore, $(x,y)\notin S_{u^*}.$
\end{itemize}
These discussions prove that $ \text{ for all $(x,y)\in S_{u^*}$, }x\in Z\setminus B_Z  \iff y\in  Z\setminus B_Z $. Therefore, $\delta_1(Z\setminus B_Z)=\delta_2(Z\setminus B_Z).$\par
Second, we will prove that
\begin{gather*}
    \delta_1(U\cap W)+\delta_1(push(U)\cap B_Z)\leq \delta_2(push(U)\cap W)+\delta_2(push(U)\cap B_Z)
\end{gather*}
 for all upper sets $U\subset B$ with respect to $B$. It suffices to prove that; for all $(x,y)\in  S_{u^*}$,
 \begin{gather*}
     x\in (U\cap W)\cup (push(U)\cap B_Z)  \Rightarrow y\in (push(U)\cap W)\cup (push(U)\cap B_Z) \text{  }
 \end{gather*}
Assume that $x\in (U\cap W)\cup (push(U)\cap B_Z) \text{ and }u^*(x)-u^*(y)=\|x-y\|_{\infty}$.
\begin{itemize}
    \item Suppose that $x\in U\cap W$. By \eqref{eq:eqivalentZBZ}, $x\in W$ implies that $y\notin Z\setminus B_Z$, which means that $y\in W\cup B_Z$. By definition of $push(U)$, $y\in push(U)$. Thus,\begin{gather*}
        y\in (push(U)\cap W)\cup (push(U)\cap B_Z).
    \end{gather*}
    \item If $x\in push(U)\cap B_Z$, $u(x)=0$. Then, $\|x-y\|_{\infty}=0$ as we have discussed earlier, and $y=x\in push(U)\cap B_Z$.
\end{itemize}
This completes the proof.   
\end{proof}

\begin{dfn*}[Definition\ref{dfn:push of function} restate]
    For a bounded function $f:B\rightarrow \mathbb{R}$, $f_{push}:X\rightarrow \mathbb{R}$ is defined as 
    \begin{gather*}
        f_{push}(x)=f(x+t\bm{1}), \text{ where } t=c+1-\max(x).
    \end{gather*}
\end{dfn*}

\begin{rmk*}[Lemma \ref{lemma:integral inequality between delta_1 and delta_2} restate]
Suppose that $c>c^* $ and there exists $\delta\in \Gamma_{+}(X\times X)$ that satisfies \eqref{eq:delta moreover cdn}, \eqref{eq:delta_1 cvx cdn}, \eqref{eq:delta_2 cvx cdn}, \eqref{eq:delta a.s. cdn} and \eqref{eq:delta integral cdn}.
Then, for all nonnegative, bounded  and nondecreasing functions $f:B \rightarrow [0, \infty)$,
\begin{gather}
\int_W fd\gamma_1=\int_W f_\text{push}d\mu_-,\tag{\ref{eq:gamma push equal}}\\
\int_W fd\delta_1\leq \int_W f_\text{push}d\delta_2+\int_{B_Z}f_\text{push}d\qty(\delta_2-\delta_1).\tag{\ref{eq:delta push inequal}}
\end{gather}
\end{rmk*}

\begin{proof}
Suppose that $b<M$ for some $M>0$. For any $m\in \mathbb{N}$, define $(U_i^m)_{i=1}^m\subset X$ as 
\begin{gather*}
    U_i^m\equiv\qty{x\in X:b(x)\geq \dfrac{M\cdot i}{m}}.
\end{gather*}
Note that $U_i^m\subset B$ is an upper set with respect to $B$ for all $m$ and $i$.
Let $b^m:B\rightarrow \mathbb{R}$ and $b_p^m:W\cup B_Z\rightarrow \mathbb{R}$ be defined as
\begin{gather*}
    b^m(x)\equiv \sum_{i=1}^m \dfrac{M}{m}\mathbb{I}(x\in U_i^m),\\
    b_p^m(x)\equiv \sum_{i=1}^m \dfrac{M}{m}\mathbb{I}(x\in push(U_i^m)).
\end{gather*}
Then, for all $m\in \mathbb{N}$,
\begin{gather*}
    \sup_{x\in B} |b^m-b|\leq \dfrac{M}{m},\\
    \sup_{x\in W\cup B_Z} |b_p^m-b_{push}|\leq \dfrac{M}{m}.
\end{gather*}
This is because for all $x\in B$, there exists $i=1\ldots, m$ such that $\frac{M}{m}(i-1)\leq b(x)< \frac{M}{m}i$. By definition of $b^m$, $b^m(x)=\frac{M}{m}(i-1)$. Hence, $|b(x)-b^m(x)|\leq \frac{M}{m} $. \par
Likewise, take any $x\in W\cup B_Z$. Since $b_{push}< M$,  there exists $i=1\ldots, m$ such that
\begin{gather*}
    \frac{M}{m}(i-1)\leq b_{push}(x)<\frac{M}{m}i.
\end{gather*} 
By definition of $b_{push}$,
\begin{gather*}
    \frac{M}{m}(i-1)\leq b(x+t\bm{1})< \frac{M}{m}i, \text{ where } t=c+1-\max(x).
\end{gather*}
We have $ x+t\bm{1}\in U^m_{i-1}\cap B$ and $x+t\bm{1}\notin U^m_{i}\cap B$.
By Lemma \ref{lemma:y+t1},
\[
x\in push(U^m_i)\setminus push(U^m_{i+1}),
\]
By definition of $b_p^m$, $b_p^m(x)=\frac{M}{m}(i-1)$ holds. Thus, we obtain $|b_{push}(x)-b_p^m(x)|\leq \frac{M}{m} $.
Therefore, $b^m$ uniformly converges to $b$ on $B$, and $b_p^m$ uniformly converges to $b_{push}$ on $W\cup B_Z$. 
First,  we will prove \eqref{eq:gamma push equal}.
Note that $\gamma_1(U)=\mu_{-}(push(U))$ for all upper set $U\subset B$ . Then, by the definition of $b^m, b_p^m$,
\begin{gather*}
     \int_B b^m d\gamma_1 =\sum_{i=1}^m \dfrac{M}{m}\gamma_1(U_i^m) 
     = \sum_{i=1}^m \dfrac{M}{m}\mu_{-}(push(U_i^m))
     =\int_{W} b^m d\mu_{-}
\end{gather*}
Hence, 
\begin{align*}
\int_W b(x)d\gamma_1 
&= \int_B b(x) d\gamma_1\quad(\because  \text{ by \eqref{eq:delta moreover cdn} and Lemma \ref{lemma:divide Z and W}, $\gamma_1$ has measure only on $Z\cup B$} )\\
&=\lim_{m\rightarrow \infty}\int_B b^m d\gamma_1 (\because \text{ $b^m$ uniformly converges to $b$}) \\
&=  \lim_{m\rightarrow  \infty}\int_{W} b_p^m d\mu_{-}\\
&=\int_{W} b_{push}(x) d\mu_{-}(\because \text{ $b_p^m$ uniformly converges to $b_{push}$})
\end{align*}
Therefore,  $\int_W b(x)d\gamma_1=\int_{W} b_{push}(x) d\mu_{-}$.
Second, we will prove
\eqref{eq:delta push inequal}.
Lemma \ref{lemma:pushu} implies that for all $m\in \mathbb{N}$
\begin{eqnarray*}
     \int_B b^m d\delta_1 
     &=& \sum_{i=1}^m \dfrac{M}{m}\delta_1(U_i^m\cap B) \\
     &=&\sum_{i=1}^m \dfrac{M}{m}\delta_1(U_i^m\cap W)  \\
     &\leq& \sum_{i=1}^m \dfrac{M}{m}\qty(\delta_2(push(U_i^m)\cap W)+(\delta_2-\delta_1)(push(U_i^m)\cap B_Z))\quad(\because \text{Lemma } \ref{lemma:pushu})\\
     &=&\int_{W} b_p^m d\delta_2+\int_{B_Z} b_p^m d(\delta_2-\delta_1)
\end{eqnarray*}
Therefore,
\begin{align*}
    \int_Wb(x)d\delta_1
    &= \int_B b(x) d\delta_1\quad(\because  \text{ by \eqref{eq:delta moreover cdn} and Lemma \ref{lemma:divide Z and W}, $\delta_1$ has measure only on $Z\cup B$} )\\
    &=\lim_{m\rightarrow \infty}\int_B b^m d\delta_1\\
    &\leq  \lim_{m\rightarrow  \infty}\qty(\int_{W} b_p^m d\delta_2+\int_{B_Z} b_p^m d(\delta_2-\delta_1) )\\
    &=\int_{W} b_{push}(x) d\delta_2+ \int_{B_Z}b_{push}(x) d(\delta_2-\delta_1).
\end{align*}

\end{proof}

\begin{rmk*}[Lemma\ref{lemma:final lemma for necessity} restate]
Suppose that $c>c^* $ and there exists $\delta\in \Gamma_{+}(X\times X)$ that satisfies \eqref{eq:delta moreover cdn}, \eqref{eq:delta_1 cvx cdn}, \eqref{eq:delta_2 cvx cdn}, \eqref{eq:delta a.s. cdn} and \eqref{eq:delta integral cdn}.
Then $\delta_1\succeq_{cvx}\mu_+$ does not hold.
\end{rmk*}

\begin{proof}
Recall we defined $v:X\rightarrow [0, \infty)$ as
\begin{gather*}
    v(x)\equiv\max\qty(\sum_{i=1}^n x_i-(nc+n+c-p), 0)
\end{gather*} 
and let $\Tilde{v}:B\to[0, \infty)$ be the restriction of $v$ to $B$, i.e. $\Tilde{v}\equiv v|_B$. Define $v_t:X\rightarrow[0, \infty)$ as
\begin{gather*}
    v_t(x)\equiv\max\qty(\Tilde{v}_{push}(x)+(n+t)(\max(x)-p),0)
\end{gather*}
for $t>0$.
Note that $\Tilde{v}_{push}:X\to [0, \infty)$ is
\begin{gather*}
    \Tilde{v}_{push}(x)=\max\qty(\sum_{i=1}^n x_i+(p-c)-n\max(x), 0)
\end{gather*}
So that
\begin{gather*}
\Tilde{v}_{push}(x)+(n+t)(\max(x)-p)
    =\max\qty(\sum_{i=1}^n x_i-c-(n-1)p,n(\max(x)-p))+t(\max(x)-p)
\end{gather*}
Since the pointwise maximum of affine functions is convex and all affine terms are non-decreasing, their sum remains convex. Thus, $v_t$ is in $\mathcal{U}(X)$ and by \eqref{eq:delta_2 cvx cdn},
\[
\int_X v_t d\qty(\mu_--\delta_2)\geq 0
\]
for all $t>0$. By $v_t(x)=\Tilde{f}_{push}(x)+(n+t)u^*(x)$ for $x\in W\cup B_Z$ and  Lemma \ref{lemma:both delta integral equal}
\begin{eqnarray*}
    0
    &\leq& \int_X v_t d\qty(\mu_--\delta_2)\\
    &=& \int_{W\cup B_Z} v_t d\qty(\mu_--\delta_2) + \int_{Z\setminus  B_Z} v_t d\qty(\mu_--\delta_2)\\
    &=&\int_{W\cup B_Z} \Tilde{v}_{push} d\qty(\mu_--\delta_2) + \int_{Z\setminus  B_Z} v_t d\qty(\mu_--\delta_2) \quad(\because \text{ Lemma }\ref{lemma:both delta integral equal})
\end{eqnarray*}
Since $\mu_{-}(B_Z)=0$ and $v\ge0$,
\begin{gather*}
    \int_{B_Z} \Tilde{v}_{push} d\qty(\mu_--\delta_2)=\int_{B_Z} \Tilde{v}_{push} d\qty(-\delta_2)\leq \int_{B_Z} \Tilde{v}_{push} d\qty(\delta_1-\delta_2)
\end{gather*}
Moreover, by \eqref{eq:gamma push equal} and \eqref{eq:delta push inequal}, we obtain
    \begin{eqnarray*}
        \int_{W\cup B_Z} \Tilde{v}_{push} d\qty(\mu_--\delta_2)
        &=& \int_{W} \Tilde{v}_{push} d\qty(\mu_--\delta_2)+\int_{B_Z} \Tilde{v}_{push} d\qty(\mu_--\delta_2)\\
        &\le& \int_{W} \Tilde{v}_{push} d\qty(\mu_--\delta_2)+\int_{B_Z} \Tilde{v}_{push} d\qty(\delta_1-\delta_2)\\
        &\leq  & \int_{W} \Tilde{v}_{push} d\qty(\mu_--\delta_2) + \int_{W} \Tilde{v}_{push} d\delta_2- \int_{W} \Tilde{v} d\delta_1 \quad(\because \eqref{eq:delta push inequal})\\
        &=& \int_{W} \Tilde{v} d\qty(\gamma_1-\delta_1) \quad(\because \eqref{eq:gamma push equal})
\end{eqnarray*}
Here, since $v(x)=0$ for $x\in Z\cup B_Z$ and both of $\gamma_1$ and $\delta_1$ have measure only on $Z\cup B$ ($\because$ \eqref{eq:delta moreover cdn} and Lemma\ref{lemma:divide Z and W}), we have
\begin{gather*}
    \int_{X} v d\qty(\gamma_1-\delta_1)=\int_Bvd\qty(\gamma_1-\delta_1)=\int_{W} \Tilde{v} d\qty(\gamma_1-\delta_1) 
\end{gather*}
Therefore, for all $t>0$,
\begin{eqnarray*}
    0\leq \int_{X} v d\qty(\gamma_1-\delta_1) + \int_{Z\setminus  B_Z} v_t d\qty(\mu_--\delta_2)
\end{eqnarray*}
By the definition of $v_t(x),$ notice that for all $x\in Z\setminus B_Z$, $v_t(x)\rightarrow 0$ as $t\rightarrow \infty$. Thus, by the bounded convergence theorem.
 \begin{gather*}
     0
     \leq \int_{X} v d\qty(\gamma_1-\delta_1)
     +\lim_{t\rightarrow \infty}\int_{Z\setminus B_Z} v_t d\qty(\mu_--\delta_2) 
     = \int_{X} v d\qty(\gamma_1-\delta_1)
 \end{gather*}
 Therefore,
 \begin{gather*}
     \int_{X} v d\gamma_1\geq \int_{X} v d\delta_1
 \end{gather*}
Since $c>c^*$ implies $(n+1)(c+1-p)<c+1$, by \eqref{eq:c>c^star implies gamma1 does not cvx dominate mu+} recall we have
\[
\int_X v d\mu_+ >\int_X vd\gamma_1.
\]
Thus, we obtain
\[
\int_X v d\mu_+ >\int_X vd\delta_1.
\]
Since $v\in \mathcal{U}(X)$,  $\delta_1\succeq_{cvx}\mu_+$ doesn't hold.

\end{proof}

By the above lemmas, Proposition \ref{prop: nescessity} is proved.

\newpage

\printbibliography

\end{document}